\documentclass[12pt]{article}
\usepackage{enumitem}
\usepackage{amsmath}
\usepackage{amsthm}
\usepackage{fancyhdr}
\usepackage{mathrsfs}
\usepackage{natbib}
\usepackage[margin=1in]{geometry}
\usepackage{color}
\usepackage{multirow,array}
\usepackage{amsmath,amsthm,amssymb}
{
\theoremstyle{plain}
\newtheorem{theorem}{Theorem}[section] 
\newtheorem{proposition}{Proposition}[section] 
\newtheorem{lemma}{Lemma}[section]
\newtheorem{assumption}{Assumption}
\newtheorem{definition}{Definition}

\newtheorem{example}{Example}
\newtheorem{algorithm}{Algorithm}
  }
\usepackage[colorlinks = true,
            linkcolor = blue,
            urlcolor  = blue,
            citecolor = blue,
            anchorcolor = blue]{hyperref}
\usepackage{hyperref}
\usepackage{verbatim}
\usepackage{tikz}
\usepackage{mathtools}

\usepackage{appendix}
\usepackage{float}
\usepackage{graphicx}
\usepackage{caption}
\usepackage{subcaption}
\usepackage[utf8]{inputenc}
\usepackage[T1]{fontenc}    
\usepackage{hyperref}       
\usepackage{url}            
\usepackage{booktabs}       
\usepackage{amsfonts}       
\usepackage{nicefrac}       
\usepackage{microtype}      

\usepackage{dsfont}
\usepackage{indentfirst}

\let\emptyset\varnothing

\numberwithin{equation}{section}

\title{Treatment Effects of Multi-Valued Treatments in Hyper-Rectangle Model}
\author{Xunkang Tian\footnote{Faculty of Doctoral Studies, European Research University. Email: \href{mailto:xunkang.tian@eruni.org}{xunkang.tian@eruni.org}}}
\date{\today}

\begin{document}

\maketitle


\begin{abstract}
This study investigates the identification of marginal treatment responses within multi-valued treatment models. Extending the hyper-rectangle model introduced by Lee and Salanié (2018), this paper relaxes restrictive assumptions, including the requirement of known treatment selection thresholds and the dependence of treatments on all unobserved heterogeneity. By incorporating an additional ranked treatment assumption, this study demonstrates that the marginal treatment responses can be identified under a broader set of conditions, either point or set identification. The framework further enables the derivation of various treatment effects from the marginal treatment responses. Additionally, this paper introduces a hypothesis testing method to evaluate the effectiveness of policies on treatment effects, enhancing its applicability to empirical policy analysis.
\end{abstract}

\begin{center}
\end{center}

\clearpage

\section{Introduction}\label{sec:intro}



In economic applications, researchers frequently encounter situations involving treatments with multiple levels. For instance, participants in labor training programs may receive varying types or intensities of training based on their individual characteristics. Similarly, in medical practice, patients may be categorized into different care groups depending on specific indicators of their health conditions. 
In such scenarios, traditional binary treatment models become inadequate, and researchers must instead rely on multi-valued treatment models. These models allow for a more accurate assessment of heterogeneous treatment effects, which is essential for effective policy evaluation and the optimal design of interventions.


However, applying these models in practice involves significant challenges. A key difficulty arises from the presence of unobserved heterogeneity, which complicates the relationship between treatment selection and outcomes. Since individuals may be grouped into treatments based on unobserved factors correlated with potential outcomes, standard methods often fail to provide reliable identification of treatment effects. 
Additionally, multi-valued treatments typically involve more complex selection mechanisms compared to binary treatments. As the number of treatment levels increases, characterizing selection behavior becomes increasingly demanding, and assumptions about treatment assignment need to be carefully justified. Traditional methods for identifying treatment effects often depend on strict assumptions such as strong functional forms, restrictive monotonicity conditions, or independence assumptions that may not hold in realistic empirical contexts. These limitations highlight the need for approaches that impose fewer and more credible restrictions while maintaining rigorous identification.


A key approach in recent literature on multi-valued treatment models involves the construction of a hyper-rectangle framework combined with instrumental variables to achieve identification. Specifically, \cite{lee2018identifying} shows how the interaction between treatment assignment thresholds and the distribution of unobserved heterogeneity determines treatment selection. They further develop methods to identify the marginal treatment response, which serves as a fundamental building block for identifying various treatment effects under multi-valued treatments. Their approach demonstrates a flexible and practical path toward identification, providing a foundation for subsequent econometric analyses.


This paper builds upon previous studies by proposing a generalized framework for identifying marginal treatment responses within multi-valued treatment models. Specifically, I extend the hyper-rectangle model introduced by \cite{lee2018identifying} with a customized set of assumptions. 
My approach considers scenarios where either the treatment selection thresholds or the distribution of unobserved heterogeneity is unknown, where one of them can be identified given the knowledge of the other one. 

To facilitate identification, I introduce a decomposition of the treatment assignment mechanism by defining a concept named the "leading term". This decomposition simplifies the analysis by clearly structuring how instrumental variables interact with unobserved heterogeneity. Further, I distinguish between cases based on the rank of these leading terms. When a leading term has full rank, the model achieves point identification of the marginal treatment responses. In contrast, if no leading term achieves full rank, I propose an additional ranked treatment assumption, enabling set identification under more realistic and flexible conditions.

Additionally, I develop a hypothesis testing method tailored explicitly to evaluate the effectiveness of policy interventions based on the identified treatment effects. This testing procedure provides policymakers with a practical tool to rigorously assess how changes in treatment assignment rules influence aggregate outcomes, thus expanding the practical applicability of multi-valued treatment models.



The literature on treatment effect identification has historically focused on binary treatment settings, where individuals are assigned either to a treatment or control group. A seminal contribution by \cite{imbens1994identification} provided conditions for identifying local average treatment effects using instrumental variables, clarifying the role of compliance behavior in treatment assignment. 
\cite{angrist1996identification} formalized the instrumental variable approach within the Rubin Causal Model, explicitly characterizing the instrumental variable estimand as the average causal effect for the subgroup of compliers. 
Subsequently, \cite{heckman1997instrumental} systematically examined the identifying assumptions necessary for estimating average treatment effects and treatment effects on the treated.

Building on these foundations, \cite{abadie2003semiparametric} developed semiparametric instrumental variable estimators to identify average treatment effects under weaker assumptions. Furthermore, \cite{imbens2004nonparametric} extended the analysis by reviewing various nonparametric methods to estimate the treatment effects and analyzing the plausibility of key assumptions. 
\cite{hahn1998role} studies the average treatment effect and average treatment effect on the treated, clarifying the role of the propensity score in efficient estimation under unconfoundedness. 
\cite{heckman2005structural} introduces the concept of marginal treatment effects to unify the nonparametric literature on treatment effects with structural econometric estimation, allowing for heterogeneity in treatment responses. 
\cite{mogstad2018using} show how instrumental variables methods can identify policy-relevant treatment parameters beyond the subpopulation directly affected by the instruments through a unified framework based on marginal treatment effects.

While much of the literature on treatment effects focuses on point identification, 
recent methodological advancements have focused on relaxing stringent assumptions associated with traditional instrumental variable frameworks, where treatment effects only admit partial identification. 
\cite{chen2023differential} propose an approach using differential treatment effects to partially identify average or heterogeneous treatment effects under unmeasured confounding, along with a two-stage inference procedure to conduct statistical inference when point identification is infeasible. 


Extending beyond the binary treatment context, recent econometric research on multi-valued treatment models captures more realistic scenarios involving multiple intervention levels. 
\cite{imbens2000role} extends the propensity score methodology from the binary treatment setting to multi-valued treatments, facilitating the estimation of average causal effects. 
\cite{cattaneo2010efficient} develops efficient semiparametric estimators for multi-valued treatment effects defined by a collection of possibly over-identified non-smooth moment conditions when the treatment assignment is under ignorability. 
Further developments by \cite{heckman2018unordered} introduce an unordered monotonicity assumption to identify treatment effects of multi-valued treatments without imposing a strict hierarchy among treatments. 
Collectively, these studies offer rigorous methodological tools and clarify essential identification issues in treatment effect models.


Beyond estimation, an important strand of the treatment effect literature focuses on hypothesis testing, particularly assessing whether a treatment has any impact and whether treatment effects vary across subpopulations. 
\cite{crump2008nonparametric} develop nonparametric tests on whether average treatment effect is zero, as well as detecting conditional average treatment effect heterogeneity across subpopulations. 
\cite{wu2021randomization} develop a version of the Fisher randomization test adapted for weak null hypotheses that do not imply sharp potential outcome restrictions. Their studentized test statistic achieves finite-sample exactness under the sharp null and retains asymptotic validity under the weak null, offering a model-free approach robust to test treatment effect heterogeneity. 
Under instrumental variable frameworks, \cite{abadie2002bootstrap} proposes a bootstrap procedure to test distributional hypotheses of treatments effects, including tests of equality of distributions and stochastic dominance. 
More recently, \cite{chernozhukov2018generic} proposes general inference methods based on machine learning proxies, facilitating estimation and testing of heterogeneous treatment effects in high-dimensional randomized experiments.


As mentioned above, it is possible that treatment effects can only be partially identified under scenarios when strong identification assumptions are relaxed. 
There is a growing body of work on hypothesis testing in set-identified frameworks. 
A foundational contribution by \cite{imbens2004confidence} proposes confidence intervals that asymptotically cover the true value of the parameter with fixed probability rather than cover the entire identified region, and its exact coverage probabilities converge uniformly to their nominal value. 
\cite{beresteanu2008asymptotic} develop a limit theory for estimators of identification regions, based on set-valued random variables and convergence in the Hausdorff metric, allowing for construction of valid confidence regions for set-identified parameters. 
\cite{romano2010inference} propose an approach to construct uniformly valid confidence regions for identified sets defined through general objective functions. 
\cite{galichon2009test} design a testing framework for non-identifying model restrictions that can be inverted to form confidence sets. Their approach complements the moment inequality-based procedures of \cite{chernozhukov2007estimation}, offering additional flexibility for hypothesis testing when model is incomplete. 
Together, these contributions provide a comprehensive foundation for conducting rigorous inference and testing in treatment effect models as well as set-identified features under a wide range of identifying assumptions.





This paper contributes to the literature on multi-valued treatment effect identification in econometrics in several ways. 
First, this paper generalizes the hyper-rectangle model introduced by \cite{lee2018identifying}, broadening its applicability. Specifically, I only require the knowledge of either the distribution of unobserved heterogeneity or treatment assignment thresholds, and consider those two scenarios. 
Additionally, my framework allows for a more flexible treatment selection mechanism, relaxing the requirement that treatment assignment must depend on all dimensions of unobserved heterogeneity. This generalized structure accommodates richer empirical settings and makes the model more applicable.

Second, I introduce the concept of leading terms to systematically analyze treatment assignment mechanisms. This approach distinguishes between scenarios where the leading term is of full rank or not, establishing clear conditions for identification. In particular, when a full-rank leading term is unavailable, I propose a novel ranked treatment assumption that achieves set identification of marginal treatment responses. This assumption aligns closely with realistic empirical settings, where higher treatment intensities typically yield systematically larger or smaller treatment effects. Thus, my framework extends the practical applicability and relevance of multi-valued treatment effect models.

Third, the paper contributes methodologically by developing a hypothesis test framework tailored to assess policy interventions' effectiveness. The approach I propose allows researchers and policymakers to evaluate rigorously whether changes in treatment assignment mechanisms yield statistically significant improvements in aggregate outcomes. This advancement bridges the gap between econometric theory and practical policy evaluation, providing a direct tool for policy analysis and decision-making.



The rest of the paper is organized as follows. 
Section~\ref{sec:model} introduces the hyper-rectangle model and formalizes the treatment assignment process. 
Section~\ref{ch:Q} presents the identification of treatment assignment threshold or distribution of unobserved heterogeneity. 
Section \ref{sec:identificationMTR} describes the strategy of identifying marginal treatment responses in different scenarios. 
Section \ref{sec:various} discusses a group of identifiable treatment effects besides marginal treatment responses. 
Section \ref{sec:test} develops the hypothesis test for evaluating the effectiveness of policies. 
Finally, Section~\ref{sec:conclusion} concludes with a discussion of implications and future directions.

\section{Hyper-Rectangle Model}\label{sec:model}

This section presents the micro-econometric model which constitutes the core analytical framework for assessing treatment assignments in this study. 
While the foundational structure of this model is inspired by and primarily conforms to \cite{lee2018identifying},  the model incorporates unique adjustments and a customized set of assumptions, which were conceived to cater to the specific context and requirements of my study.

\subsection{Model Basics}

To provide a brief overview, a standard treatment model encompasses an observed outcome, denoted by $Y \in \mathds{R}$, and the observed treatment, symbolized as $D=1,2,...,T$. Accompanying these, we have observed covariates, denoted by $X \in \mathds{R}^{\bar B}$, where $\bar B \in \mathds{N}^+$ signifies the dimension of $X$. Complementing these entity is a random vector $V=(V_1,...,V_J) \in \mathds{R}^J$ that accounts for unobserved heterogeneity. In this context, $J \in \mathds{N}^+$ represents the dimension of $V$. 
I restrict the outcome $Y$ to be strictly positive and bounded above, that is, $0<Y<\bar Y$ for some $\bar Y\in \mathds{R}^+$. I also confine the unobserved heterogeneity $V$ to the interval $[0,1]^J$. These constraints simplify the ensuing mathematics without significantly undermining the generality of the model. 

Additionally, for each outcome, we observe an instrumental variable, denoted by $Z=(Z_1,...,Z_{\bar{W}}) \in \mathds{R}^{\bar{W}}$, where ${\bar{W}} \in \mathds{N}^+$ is the dimension of $Z$. These instruments play a crucial role in identification and estimation strategies  in the later analysis.

The observed data is composed of 
a sample $\{ (Y^o,D^o,Z^o,X^o): o=1,...,N_o \}$, with $N_o \in \mathds{N}^+$ denoting the sample size. For the sake of notational simplicity, I suppress the conditioning on $X$ in subsequent discussions, and all results should be interpreted as conditional on $X$.

To formalize the relationship between the observed outcome and treatment, let $Y_k$, $k=1,...,T$ denote the potential outcome under treatment $k$, and let $D_k \equiv \mathds{1}\{D=k\}$, $k=1,...,T$. The observed outcome $Y$ can thus be expressed as a sum of potential outcomes weighted by their respective treatment indicators: $Y=\sum\limits_{k=1}^T Y_kD_k$.

The legitimacy of the chosen instruments $Z$ is substantiated through the following assumption:

\begin{assumption}[Conditional Independence]\label{assum:independence}
  The potential outcomes $Y_k$, $k=1,...,J$, and the unobserved heterogeneity $V$ are jointly independent of the instruments $Z$.
\end{assumption}

The objective of this analysis is to estimate the Marginal Treatment Response (MTR), defined as $E[Y_k|V=v]$, which captures the expected potential outcome $Y_k$ given a specific realization of the unobserved heterogeneity $V=v$. I impose the continuity of the MTR in the Data Generating Process (DGP):

\begin{assumption}[Local Equicontinuity]\label{assum:continuousmtr}
  The set of Marginal Treatment Response function $\left\{ E[Y_k|V=v] \right\}_{k=1}^T$ is locally equicontinuous at each $v \in (0,1)^J$.
\end{assumption}

This continuity assumption ensures the mathematical tractability of the model and allows us to employ a host of econometric techniques to analyze the data.

\subsection{The Determination of Treatment}\label{ch:determoftreat}

Next, I delve into the mechanism determining the treatment variable $D$. This is controlled by the confluence of a series of conditions. Specifically, the conditions entail a set of inequalities involving unobserved heterogeneity $V$: $V_1<Q_1(Z)$ or $V_1 \geq Q_1(Z)$, ..., $V_J<Q_J(Z)$ or $V_J \geq Q_J(Z)$. Here, $Q(Z)=(Q_1(Z),...,Q_J(Z))$ represents a vector of functions of the instrumental variable $Z$, and it acts as the threshold for $V$. 
I impose a key assumption about the support of the threshold, which I restrict to be the open interval $(0,1)$. This assumption is formalized as follows:
\begin{assumption}[Interior of Threshold]\label{interiorofQ}
Let $\mathcal{Z}$ be the support of instrument $Z$. Then, 
for all $j=1,\ldots,J$, there is $Q_j(\mathcal{Z})=(0,1)$.
\end{assumption}
Assumption \ref{interiorofQ} precludes uninteresting scenarios in which the threshold $Q_j(Z)$ reaches a boundary point, rendering an explicit threshold ineffective in influencing the treatment assignment. 
In other words, it ensures that all realizations of $Q(Z)$ lie within the interior of its range of variation. 
Besides, the open interval of $Q(\mathcal{Z})$ also indicates any realization of $Q(Z)$ belongs to the interior of its range of variation. 
Furthermore, the assumption that the support of $Q(Z)$ is dense in $(0,1)$ guarantees that changes in the instrument $Z$ can bring about the entire spectrum of variation in the threshold $Q(Z)$. This is critical for the identification analysis to follow, as it ensures that there is sufficient exogenous variation in the instrument to trace out the treatment effect of interest.
In the following discussions, I will separately consider the cases where the threshold function $Q(Z)$ is known and unknown.

In this model, treatment $D=k$ is selected if and only if the specified inequalities are satisfied. To formalize this, consider the $\sigma$-algebra $\sigma_{\{ V,Q(Z) \}}$ generated by the set
\begin{equation}\label{sigmaalgebra}
  \{ V_j<Q_j(Z) \}, \quad  j=1,..., J
\end{equation}
I put forward the following assumption
\begin{assumption}[Measurability]\label{assum:measure}
 Treatment variable $D$ is measurable with respect to the $\sigma$-algebra  $\sigma_{\{ V,Q(Z) \}}$.
\end{assumption}

Any set in the $\sigma$-algebra corresponds to taking unions, intersections, and complementation of  the sets in (\ref{sigmaalgebra}). Therefore,
we can envision the treatment model as being constructed on a hyperplane, where every $V_j$ forms one dimension. Given that $V_j\in [0,1]$ and is partitioned by $Q_j(Z)$, 
each hyper-rectangle is formed from the intersection of chosen $V_j<Q_j(Z)$ or $V_j\geq Q_j(Z)$, $j=1,...,J$. 
Each treatment is linked to a combination of one or several hyper-rectangles in this hyperplane.  
The following example demonstrates this style of treatment determination, providing a more concrete understanding of the framework presented above.

\begin{example}[Multi-Way Selection]\label{example1}
  Consider a comprehensive training program. This program evaluates participants based on their performance in three distinct tests and subsequently assigns them into one of the four designated groups, represented as $D=1,2,3,4$. Denote the scores of the participants in these tests as $V_1$, $V_2$, $V_3$, and the minimum required grades for passing each test as $Q_1(Z)$, $Q_2(Z)$, $Q_3(Z)$. In this scenario, $Z$ refers to certain characteristics or conditions that may influence the required threshold for qualification.

  The process of group assignment is primarily determined by the following set of rules:
  \begin{itemize}
    \item Participants who fail to meet the thresholds for tests 2 and 3, irrespective of their performance in test 1, are assigned to group 1.
    \item Participants are assigned to group 2 if they successfully pass test 1 and only one test among tests 2 and 3.
    \item Participants who excel by surpassing the thresholds in all three tests are assigned to group 3.
    \item Group 4 comprises participants who fail test 1, but manage to pass at least one test among tests 2 and 3.
  \end{itemize}

  Aligning this example with the framework of my model, we see that $J=3$ and $T=4$. Thus, the treatment $D$ is ascertained by the following set of conditions:
  \begin{itemize}
    \item $D=1$ if $V_2<Q_2(Z)$, $V_3<Q_3(Z)$.
    \item $D=2$ if $V_1\geq Q_1(Z)$, $V_2<Q_2(Z)$, $V_3\geq Q_3(Z)$, or $V_1\geq Q_1(Z)$, $V_2\geq Q_2(Z)$, $V_3< Q_3(Z)$.
    \item $D=3$ if $V_1\geq Q_1(Z)$, $V_2\geq Q_2(Z)$, $V_3\geq Q_3(Z)$.
    \item $D=4$ if $V_1< Q_1(Z)$, $V_2<Q_2(Z)$, $V_3\geq Q_3(Z)$, or $V_1< Q_1(Z)$, $V_2\geq Q_2(Z)$.
  \end{itemize}

This case is demonstrated graphically in Figure \ref{J=3example}. Here, the hypercube $[0,1]^J$ is sectioned into eight distinct regions. Each region corresponds to a specific combination of test outcomes and leads to a unique treatment assignment.
\begin{figure}[htbp]
  \centering
  \hspace{-2cm}
  \begin{minipage}[b]{0.25\textwidth}
    \includegraphics[width=\linewidth]{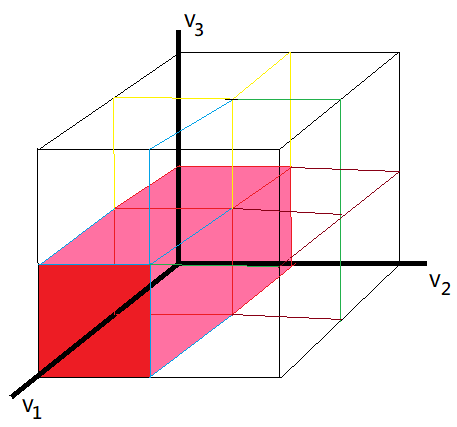}
    \centering (a) $D=1$
  \end{minipage}
  \hspace{0cm}
  \begin{minipage}[b]{0.25\textwidth}
    \includegraphics[width=\linewidth]{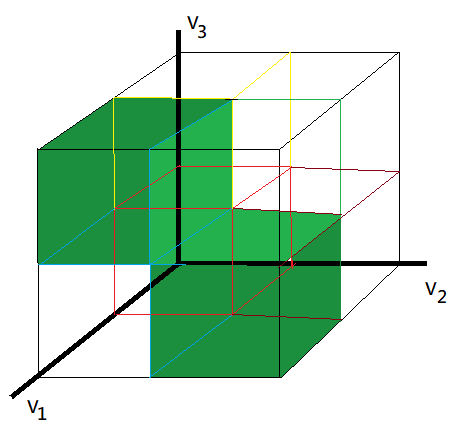}
    \centering (b) $D=2$
  \end{minipage}
  \hspace{0cm}
  \begin{minipage}[b]{0.25\textwidth}
    \includegraphics[width=\linewidth]{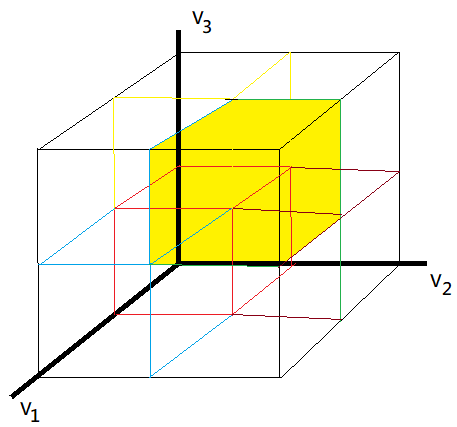}
    \centering (c) $D=3$
  \end{minipage}
  \hspace{0cm}
  \begin{minipage}[b]{0.25\textwidth}
    \includegraphics[width=\linewidth]{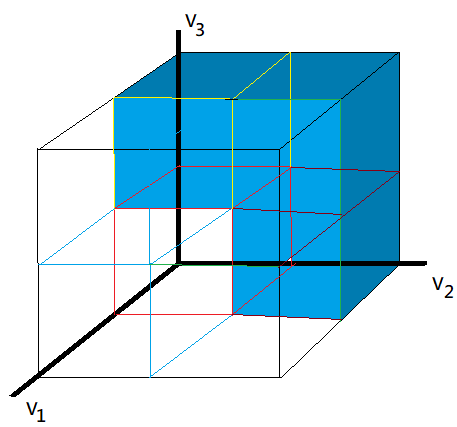}
    \centering (d) $D=4$
  \end{minipage}
  \hspace{-2cm}
  \caption{Illustration of Example \ref{example1}}
  \label{J=3example}
\end{figure}
\end{example}

To precisely illustrate the treatment selection process, I introduce a function $d_k(V,Q(Z))$ such that the treatment indicator can be represented as $\mathds{1}\{ D=k \}=d_k(V,Q(Z))$. 
For the sake of simplification,  denote
\begin{equation}\label{eq:Sj}
  S_j(V,Q(Z)) \equiv \mathds{1}\{ V_j<Q_j(Z) \}, \quad j=1,..., J
\end{equation}
Given that $\mathds{1}\{ V_j \geq Q_j(Z) \} = 1- \mathds{1}\{ V_j < Q_j(Z) \}$, it is observed that for each treatment $k$, the indicator function $d_k(V,Q(Z))$ equates to the summation, product, and difference of selected $\mathds{1}\{ V_j<Q_j(Z) \}$, $j=1,..., J$, represented in a polynomial form.

Now, let $\mathcal{L}$ denote the set of all non-empty subsets $l$ of $\mathbf{J} \equiv \{ 1,..,J \}$. In this context, $d_k(V,Q(Z))$ can be articulated according to how the hyper-rectangle for treatment $k$ is constructed. Mathematically, it is expressed as
\begin{equation}\label{rectangledk}
  d_k(V,Q(Z))= \sum_{l\in \mathcal{L}}  e_l^k  \prod_{j\in l}  S_j(V,Q(Z))^{r_{lj}^k} (1-S_j(V,Q(Z)))^{1-r_{lj}^k}
\end{equation}
where $e_l^k \in \{0,1\}$ signifies the existence of term $l$ in the set for treatment $k$, and $r_{lj}^k \in \{0,1\}$ demonstrates whether $V_j<Q_j(Z)$ or $V_j\geq Q_j(Z)$ is involved.
Upon polynomial expansion, $d_k(V,Q(Z))$ can be distinctly expressed in a decomposed form as
\begin{equation}\label{decomposedk}
  d_k(V,Q(Z))= \sum_{l\in \mathcal{L}} c_l^k \prod_{j\in l}  S_j(V,Q(Z))
\end{equation}
where $c_l^k \in \mathds{Z}$ denotes the integer coefficient of the term $l$ for treatment $k$.

Reverting to Example \ref{example1},  we can lucidly express each treatment in accordance with Equation \eqref{rectangledk} and \eqref{decomposedk} respectively.
\begin{eqnarray}\label{example1equations}
  d_1 &=& S_2S_3 \nonumber \\
  d_2 &=& (1-S_1)S_2(1-S_3)+(1-S_1)(1-S_2)S_3 \nonumber \\
   &=& S_2+S_3-S_1S_2-S_1S_3-2S_2S_3+2S_1S_2S_3 \nonumber \\
  d_3 &=& (1-S_1)(1-S_2)(1-S_3) \nonumber \\
   &=& 1-S_1-S_2-S_3+S_1S_2+S_1S_3+S_2S_3-S_1S_2S_3 \nonumber \\
  d_4 &=& S_1S_2(1-S_3)+ S_1(1-S_2) \nonumber \\
   &=& S_1-S_1S_2S_3
\end{eqnarray}

In addition, it is assumed that each realization of $(V,Z)$ must be associated with one and only one treatment. Formally, this can be expressed as follows:

\begin{assumption}[Completeness]  \label{assum:complete}
For any particular instance of $V$ and $Z$, it holds that
\begin{equation}\label{equation:complete}
  \sum_{k=1}^T d_k(V,Q(Z)) =1
\end{equation}
\end{assumption}

This assumption reinforces the exclusivity of the treatment assignment, precluding any void or overlap in treatments for any given combination of $V$ and $Z$. Moreover, it ensures that the expression given in Equation \eqref{decomposedk} is uniquely determined. Therefore, the treatment determination mechanism can be expressed as a function:
\begin{equation}\label{eq:detofd}
  d(V,Q(Z)) = \sum_{k=1}^T k \cdot d_k(V,Q(Z))
\end{equation}

Furthermore, I introduce the following assumption: 
\begin{assumption}[Involvement of Threshold] \label{assumption:threshold_involvement}
Each threshold $Q_j(Z)$ must be involved in at least one treatment's determination mechanism $d_k(V,Q(Z))$, $k=1,...,T$. Formally, there does not exist a $j \in \mathbf{J}$ such that $j \notin l$ for all $l$ in $\{ l: \exists k \in \{ 1,...,T \}, \text{ s.t. } c_l^k \neq 0 \}$.
\end{assumption}

Assumption \ref{assumption:threshold_involvement} essentially states that we preclude the possibility of any dimension $j$ that does not contribute to the diversity of treatment assignments. This assumption is made without loss of any heterogeneity in our model. If a certain dimension of unobserved heterogeneity is not influential in determining a treatment, it can be eliminated from both $V$ and $Q(Z)$ without fundamentally altering the model. In doing so, we reduce the dimensionality from $J$ to $J-1$. This reduction simplifies the mathematical manipulations and clarifies the interpretations of our model.

An important aspect to note is that any $V_j$, $V_j'$ ($j\neq j'$) can be correlated or can even be identical. This implication is significant as the hyper-rectangle model can seamlessly incorporate the truncated model for treatment determination. Specifically, for a one-dimensional variable $V\in [0,1]$, it is possible to have two or more truncation points that divide the uniform interval into three or more segments, resulting in multi-valued treatments. 

Figure \ref{V1V6} provides an example with $J=4$, where $V_1=V_2$ and $V_3=V_4$. The scenario can then be represented in a two-dimensional plane instead of a four-dimensional space.

\begin{figure}[htbp]
  \centering
    \includegraphics[width=0.5\linewidth]{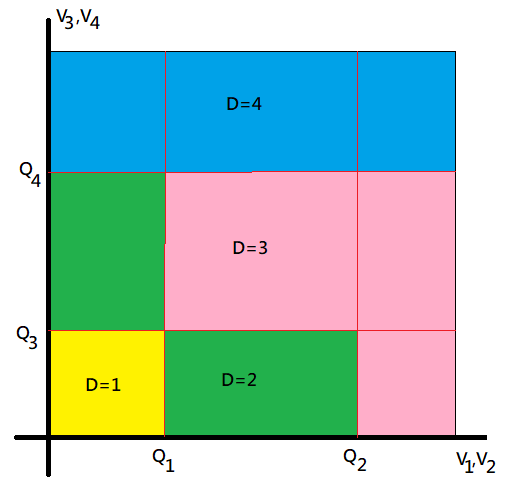}
  \caption{Example with $V_1=V_2$, $V_3=V_4$}
  \label{V1V6}
\end{figure}

However, in order to incorporate this scenario within the framework of my model, it is necessary to introduce an "outside treatment option" $D=5$ to account for all cases not captured by this plane, such as $V_1<Q_1(Z)$, $V_2\geq Q_2(Z)$. Thus, in the scenario depicted by Figure \ref{V1V6}, the total number of treatments should be $T=5$.  
\footnote{
The classical unconfoundedness assumption, $(Y_k)_{k=1}^J \perp\!\!\!\perp D \mid X$, is not required in this framework. 
The unobserved heterogeneity $V$ is correlated with $Y_k$, and also enters the treatment assignment rule since $D$ is a deterministic function of $(V,Z)$. Conditioning on $X$ and $Z$ does not break the dependence between $D$ and $Y_k$. 
}

\subsection{Leading Term}

To facilitate the estimation process, I will employ the term "term $l$" to denote each $c_l^k \prod_{j\in l}  S_j(V,Q(Z))$ with $c_l^k \neq 0$ in Equation \eqref{decomposedk}. Consequently, $d_k(V,Q(Z))$ can be viewed as a finite combination of the term $l\in \mathcal{L}$.

Let $l_j=\mathds{1}\{ j\in l \}$, a term can be succinctly represented by a vector with coefficient as $l=c_l^k (l_1,...,l_J)$. For instance, in Example \ref{example1}, the $D=1$ case implies only one term $l=(0,1,1)$, and the $D=2$ case illustrates six terms, namely, $l^1=(0,1,0)$, $l^2=(0,0,1)$, $l^3=-(1,1,0)$, $l^4=-(1,0,1)$, $l^5=-2(0,1,1)$, and $l^6=2(1,1,1)$, as implied by Equation \eqref{example1equations}.
To further clarify, I propose the following definitions related to the term:

\begin{definition}[Inclusion]
A term $l$ is said to be included in another term $l'$ (denoted as $l \subseteq l'$), if $l'$ encompasses the $S_j$ that are present in $l$. In other words, $l_j \leq l_j'$, $\forall j=1,...,J$. This holds regardless of the values or signs of the coefficients of these terms.
\end{definition}

\begin{definition}[Rank]
The rank of a term $l$, symbolized as $|l|$, is determined by the count of its non-zero elements. Formally, it can be expressed as:
\begin{equation*}
  |l|=\sum_{j=1}^J  l_j
\end{equation*}
A term $|l|$ with a rank equivalent to $J$ is said to be of full rank.
\end{definition}

\begin{definition}[Leading Term]
In a treatment decomposition, a term $l^i$ is called a leading term if there exists no other term $l^{i'}\neq l^i$ such that $l^i \subseteq l^{i'}$.
\end{definition}

I revisit Example \ref{example1} to illustrate the point of leading term. 
In the case of $D=1$, the only term present is the leading term, which holds a rank of $2$. In the case of $D=2$, there is a sole leading term $l^6$, with a rank of $3$. From the definitions, some noteworthy propositions can be yielded, which shed light on the fundamental characteristics of treatment decompositions. 

\begin{proposition}
  In a treatment decomposition, there may exist one or multiple leading term.
\end{proposition}

\begin{proposition}\label{prop2}
  In a treatment decomposition, if term $l^i$ possesses full rank, it is the unique leading term.
\end{proposition}

\begin{proposition}\label{prop3}
  In a treatment decomposition, if term $l^i$ is a leading term, for any other term $l^{i'}$, there always exists an index $j\in \mathbf{J}$, such that $j\in l^i$ while $j\notin l^{i'}$.
\end{proposition}

Notably, the rank of the leading term could be less than the number of dimensions involved in treatment determination. An illustration of this can be found in the treatment scenario depicted in Figure \ref{leadingrank2}.

\begin{figure}[htbp]
  \centering
    \includegraphics[width=0.33\linewidth]{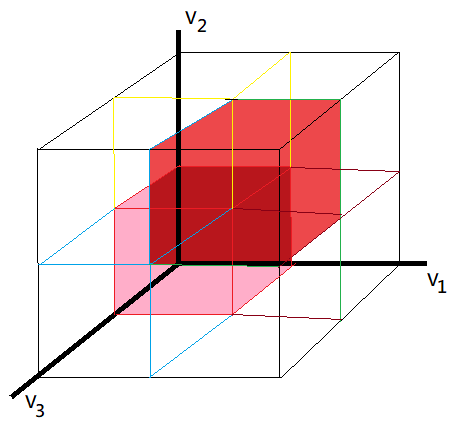}
  \caption{Leading Term with Rank 2}
  \label{leadingrank2}
\end{figure}

The corresponding analytical representation of the treatment is:
\begin{eqnarray*}
  d &=& S_1 S_2 S_3 + (1-S_1)(1-S_2)(1-S_3) \\
   &=& 1-S_1-S_2-S_3+S_1 S_2 +S_1 S_3 +S_2 S_3
\end{eqnarray*}

This expression implies three leading terms, $(1,1,0)$, $(1,0,1)$, and $(0,1,1)$. Interestingly, the rank of each of these leading terms is only two, instead of three. 
Intuitively, this reduction in complexity is attributed to the perfect predictability of one rectangle from another in this particular treatment determination mechanism, effectively reducing the degrees of freedom in treatment assignments. This feature underlines the flexibility of our model and its ability to accommodate a variety of treatment determination mechanisms.

\section{Identifying Threshold or Distribution of Heterogeneity}\label{ch:Q}

In Section \ref{ch:determoftreat}, I introduced the threshold function $Q(Z)$, the knowledge status of which, known or unknown, significantly influences the strategies for empirical analysis. This section addresses each case in turn, and provide a methodology for identifying the threshold function, setting the stage for subsequent analyses.

\subsection{Known Threshold, Unknown Distribution of Heterogeneity}

I first turn to the scenario where the threshold function $Q(Z)$ is explicitly known. For illustrative purposes, consider Example \ref{example1}, in which the minimum grades required for passing the tests are clearly specified in the program descriptions. While these thresholds may differ depending on the candidate's attributes, they are fully observable to the researcher.

The explicit knowledge of the threshold function bestows significant analytical flexibility. Importantly, it eliminates the necessity for assuming prior knowledge of the distribution of unobserved heterogeneity. Instead, we can identify this distribution directly from the data, negating the need to pre-suppose it as a known entity. This approach simplifies the empirical model and bolsters its tractability. The nuances of this advantage and its subsequent implications will be comprehensively discussed in the following discourse.

Our analysis begins with a fundamental assumption regarding the distribution of the unobserved heterogeneity $V$. We denote the probability distribution function of this unobserved heterogeneity as $f_V$, and its cumulative distribution function as $F_V$. In order to discern $f_V$ within our model, we necessitate the following condition:
\begin{assumption}[Continuity of Distribution]\label{continuousfV}
  The probability distribution function $f_V(v)$ is positive and continuous at each $v \in (0,1)^J$.
\end{assumption}

For any treatment $k$, 
the conditional probability of treatment assignment, denoted as $\Pr(D=k|Q(Z)=q)$, can be directly observed from the data as a function of $q$, where $q=(q_1,...,q_J)$ constitutes the realization of $Q(Z)$. Additionally, the following relationship holds:
\begin{eqnarray*}
  \Pr(D=k|Q(Z)=q) &=& \Pr( d_k(V,Q(Z))=1 |Q(Z)=q)  \\
   &=& \Pr( d_k(V,q)=1 )  \\
   &=& \int \mathds{1}\{ d_k(v,q)=1 \} f_V(v) dv,
\end{eqnarray*}
where the second equality is justified by the conditional joint independence of $V$ and $Z$ stipulated in Assumption \ref{assum:independence}. Given that $d_k$ is an indicator function, $\mathds{1}\{ d_k(v,q)=1 \} = d_k(v,q)$. Hence, with respect to Equation \eqref{decomposedk}, the conditional probability can be reformulated as:
\begin{equation}\label{decompprob}
  \Pr(D=k|Q(Z)=q) = \int \sum_{l\in \mathcal{L}} c_l^k \prod_{j\in l} S_j(v,q) f_V(v) dv.
\end{equation}

Equation \eqref{decompprob} provides a decomposition of the conditional probability of being assigned to treatment $k$ into its constituent terms. However, the existence of a treatment with a full rank leading term, versus every treatment's leading term being non-full rank, will significantly impact subsequent analysis. 
The upcoming content of this section will explore them in greater details.

\subsubsection{Existence of Treatment with Full Rank Leading Term}
\label{sec:fullrankleadingfv}

A common and particularly manageable case arises when there exists at least one treatment with a full rank leading term. In such scenarios, researchers can directly identify the distribution of the unobserved heterogeneity without additional stringent assumptions. This specific case is extensively discussed by \cite{lee2018identifying}, and their results can be directly applied to our context.

In Equation \eqref{decompprob}, the behavior of $S_j$ as an indicator function for the threshold, as implied by Equation \eqref{eq:Sj}, implies that $q$ linearly modulates the area of integration. Furthermore, Assumption \ref{continuousfV} establishes the continuity of $f_V$. Assumption \ref{interiorofQ}, in turn, ensures that Equation \eqref{decompprob} is well defined within an open neighborhood of $q$ for all $q\in (0,1)^J$. Consequently, we can infer that Equation \eqref{decompprob} is differentiable across all dimensions of $q$.\footnote{For further proof details, please refer to Section 6 of \cite{lee2018identifying}.}

Without loss of generality, I suppose it is treatment $k$ that has a full rank leading term. 
Denote this full rank leading term of as $\tilde{l}$. A differentiation of $\Pr(D=k|Q(Z)=q)$ with respect to all the dimensions of $q=(q_1,...,q_J)$ will render all terms outside $\tilde{l}$ null as they lack at least one dimension of $q$. Thus, we obtain:
\begin{eqnarray*}
  \frac{\partial^J}{\partial q}\Pr(D=k|Q(Z)=q) &=&  \frac{\partial^J}{\partial q}  \int c_{\tilde{l}}^k \prod_{j\in \tilde{l}}  S_j(v,q)    f_V(v) dv  \\
  &=& \frac{\partial^J}{\partial q_1 \cdots \partial q_J} \int_0^{q_1} \cdots \int_0^{q_J} c_{\tilde{l}}^k  f_V(v_1,...,v_J) dv_1\cdots dv_J   \\
  &=& c_{\tilde{l}}^k f_V(q_1,...,q_J)
\end{eqnarray*}

From this derivative, we obtain:
\begin{equation*}
  f_V(v_1,...,v_J)= \frac{1}{c_{\tilde{l}}^k} \left. \frac{\partial^J \Pr(D=k|Q(Z)=q)}{\partial q} \right|_{q=v}   
\end{equation*}
which identifies the probability density function $f_V(v)$ almost everywhere throughout its entire support, $v\in(0,1)^J$.

\subsubsection{Absence of Treatment with Full Rank Leading Term}
\label{sec:nofullrankleadingfv}

In scenarios where a full rank leading term is not present on any treatment in the model, the direct identification of $f_V$ from any specific treatment exposure becomes infeasible. Instead, by emulating the analysis structure in Section \ref{sec:fullrankleadingfv}, we can learn certain marginal distribution functions via differentiation.

Consider a leading term $l$ for treatment $k$. Let $i_1,...,i_{|l|}$ be the numbering of $l$'s elements which are equal to one, that is, $l_{j}=1$ for $j\in \{ i_1,...,i_{|l|} \}$ and $l_{j}=0$ for the other cases. We denote $\{ -i_1,...,-i_{J-|l|} \}$ as the complement set of $\{ i_1,...,i_{|l|} \}$ in $\{1,...,J\}$. 

For simplicity in notation, we introduce the following sets:
\begin{eqnarray*}
    I_l^+ &\equiv& \{ i_1,...,i_{|l|} \} \\
    I_l^- &\equiv& \{ -i_1,...,-i_{J-|l|} \}
\end{eqnarray*}

As per Equation \eqref{decompprob}, because $l$ is a leading term, differentiating with respect to dimensions in $I_l^+$ will cancel all terms other than $l$ in the summation over $l\in \mathcal{L}$, as implied by Proporsition \ref{prop3}. Thus, we find:
\begin{equation}\label{distributionidentification}
  \int_0^1 \cdots \int_0^1  f_V(v_1,...,v_J) dv_{-i_1} \cdots dv_{-i_{J-|l|}}  = \frac{1}{c_{l}^k} \left. \frac{\partial^{|l|} \Pr(D=k|Q(Z)=q)}{\partial q_{i_1}\cdots \partial q_{i_{|l|}} } \right|_{q=v}
\end{equation}

We refer to Equation \eqref{distributionidentification} as the distribution identification equation, as it provides us with a marginal distribution of $f_V(v)$ on the dimension $I_l^+$. 
In the pursuit of further simplicity in notation, we introduce 
$v_{I_l^-} \equiv (v_{-i_1}, \cdots, v_{-i_{J-|l|}})$ and $ q_{I_l^+} \equiv  (q_{i_1}, \cdots,  q_{i_{|l|}})$. 
Therefore, the left hand side of Equation \eqref{distributionidentification} can be re-expressed in a concise manner as
\begin{equation}   \label{marginalPDF}
   f_{I_l^+}  \equiv \int f_V(v) dv_{I_{l^1}^-}
\end{equation}
Here, $f_{I_l^+}$ is a notation I introduce to encapsulate this marginal probability density function.

By integrating the distribution identification equations from all the leading terms of all treatments, we can discern new insights. Suppose we have two leading terms, $l$ and $l'$, where $l \subseteq l'$. It follows from the definition of leading terms that they must originate from distinct treatments. The distribution identification equation of term $l$ is consequently implied by that of term $l'$, as long as we undertake an additional integration over the dimensions in $I_l'^+ \backslash I_l^+$. This implies that the distribution identification equation of term $l$ can be disregarded as it fails to provide new information. 


Upon gathering all the distribution identification equations that have not been eliminated, we construct an equations system, which encapsulates the accessible information concerning the distribution of $V$. Let these terms be $l^1, ..., l^P$, where $P\in \mathds{Z}^+$ is the number of those terms, the system of equations can be represented as:
\begin{eqnarray*}
    \int f_V(v) dv_{I_{l^1}^-} &=&  \frac{1}{c_{l^1}^{k^1}} \left. \frac{\partial^J \Pr(D=k^1|Q(Z)=q)}{\partial q_{I_{l^1}^+} } \right|_{q=v} \\
     &\vdots &   \\
    \int f_V(v) dv_{I_{l^P}^-} &=&  \frac{1}{c_{l^P}^{k^P}} \left. \frac{\partial^J \Pr(D=k^P|Q(Z)=q)}{\partial q_{I_{l^P}^+} } \right|_{q=v}
\end{eqnarray*}
Here, $k^p$ corresponds to the treatment associated with term $l^p$, where $p=1,...,P$. 
In the next step, we will infer the joint distribution $f_V(v)$ with the information provided by these marginal distributions. 
Assumption \ref{assumption:threshold_involvement} guarantees  the inclusion  of all dimensions $j=1,...,J$ in the combination of indices  ${I_{l^p}^+}$, $p=1,...,P$. Consequently, as stated by Sklar's theorem \citep{sklar1959fonctions}, there exists a copula $C$ that satisfies the following relation:
\begin{equation*}
    F_V(v)=C\left(F_{I_{l^1}^+}(v_{I_{l^1}^+}),...,F_{I_{l^P}^+}(v_{I_{l^P}^+}); \beta_C \right)
\end{equation*}
Here, $\beta_C$ refers to the copula's coefficients, which may have finite or infinite dimension, and  $F_{I_{l}^+}$ is the cumulative distribution of the marginal probability density $f_{I_l^+}$. 
Furthermore,  Assumption \ref{continuousfV} imposes the continuity of $f_V(v)$, which ensures the copula is uniquely determined by Sklar's theorem. Therefore, the distribution $f_V(v)$ is identified across its support $v\in (0,1)^J$.

In order to establish $F_V(v)$ from the copula, we suggest an implementation which consists of the following steps.

\begin{enumerate}[label=Step \arabic*:, align=left]
  \item Compute the marginal cumulative distribution $F_{I_{l^p}^+}$ from the marginal probability density $f_{I_{l^p}^+}$, $p=1,...,P$. 
  \item Selecting an appropriate copula function $\bar{C}(\cdot;\beta_{\bar{C}})$ which maps the marginal distributions $F_{I_{l^p}^+}$, $p=1,...,P$ to the joint distribution $F_V(v)$. It is crucial to note that the coefficients $\beta_{\bar{C}}$ remain to be ascertained. Commonly used copula functions include the Gaussian Copula, Frank Copula, and Gumbel Copula. \label{copula:steptwo}
  
  \item Construct a grid of $V$ and utilize the chosen copula to compute the realized joint distribution $\hat{F}_V(v; \beta_{\bar{C}}) =\bar{C}\left(F_{I_{l^1}^+}(v_{I_{l^1}^+}),...,F_{I_{l^P}^+}(v_{I_{l^P}^+}); \beta_{\bar{C}} \right)$ as a function of $\beta_{\bar{C}}$. Thus the corresponding joint probability density is 
  \begin{equation*}
      \hat{f}_V(v; \beta_{\bar{C}}) =\bar{C}\left(F_{I_{l^1}^+}(v_{I_{l^1}^+}),...,F_{I_{l^P}^+}(v_{I_{l^P}^+}); \beta_{\bar{C}} \right) \prod_{p=1}^P  F_{I_{l^p}^+}(v_{I_{l^p}^+})
  \end{equation*}

  \item Generate a sample of observations from the marginal distribution $F_{I_{l^1}^+}(v_{I_{l^1}^+}),...,F_{I_{l^P}^+}(v_{I_{l^P}^+})$. It is important to recognize that a single $v_j$ may appear in multiple marginal distributions. Consequently, in each sampling trial, we propose to:
  
  \begin{itemize}[align=left]
    \item Randomly select a marginal distribution $F_{I_{l^p}^+}(v_{I_{l^p}^+})$ and sample $v_{I_{l^p}^+}$ from it.
    
    \item Select another marginal distribution $F_{I_{l^{p'}}^+}(v_{I_{l^{p'}}^+})$ at random, establish its conditional distribution on $v_{I_{l^p}^+}$, denoted as $F_{I_{l^{p'}}^+|I_{l^{p}}^+}(v_{I_{l^{p'}}^+}|v_{I_{l^{p}}^+})$, if $I_{l^{p}}^+ \cap I_{l^{p'}}^+ \neq \emptyset$. Subsequently, sample $v_{I_{l^{p'}}^+}\backslash v_{I_{l^p}^+}$ from it.
    
    \item Iterate the previous steps until all elements in $v$ are sampled.
  \end{itemize}
  
  Denote the set of samples as $\{ v_{1g},...,v_{Jg} \}_{g=1}^G$, where $G$ is the predetermined sample size.
  
  \item Construct the likelihood function as
    \begin{equation*}
        \mathcal{L}(\beta_{\bar{C}})=\prod_{g=1}^G \hat{f}_V(v_g ; \beta_{\bar{C}})
    \end{equation*}
    where $v_g=( v_{1g},...,v_{Jg} )$. By employing maximum likelihood estimation, we can obtain estimates for $\beta_{\bar{C}}$, denoted as $\hat{\beta}_{\bar{C}}$. Substitute these values back into $\hat{F}_V(v; \beta_{\bar{C}})$ and $\hat{f}_V(v; \beta_{\bar{C}})$  to arrive at the estimates for the joint distribution.
\end{enumerate}

Given the intrinsic properties of convergence of maximum likelihood estimator, we have the following theorem.
\begin{theorem}\label{th:copulaconverge}
  Under Assumption \ref{assum:independence} - \ref{assumption:threshold_involvement} and \ref{continuousfV}, 
  if the copula function $\bar{C}(\cdot;\beta_{\bar{C}})$ in \ref{copula:steptwo} is correctly specified, the maximum likelihood estimates $\hat{\beta}_{\bar{C}}$ converges to the true value $\beta_{\bar{C}}$ as the number of sampling $G$ is large enough.
\end{theorem}

As implied by Theorem \ref{th:copulaconverge}, it is possible to recover the probability distribution of $V$ even in the absence of a full rank leading term.

\subsection{Unknown Threshold, Known Distribution of Heterogeneity}

The precise specification of the threshold function is frequently beyond the grasp of researchers, a phenomenon attributable to a multitude of factors. For instance, the allocation of treatment assignment may be underpinned by a latent group formation process that takes place behind the scenes, such as the influence of social networks. In such circumstances, individuals are implicitly categorised into distinct groups based on unobserved characteristics or shared experiences, resulting in a latent group structure. It is also plausible that the threshold function is kept undisclosed and thus unobservable. For instance, in the context of a training program, the organiser may conceal their exact categorisation criteria from researchers. Under such circumstances, the identification of $Q(Z)$ becomes crucial for further analysis. 

To identify the threshold, knowledge of the distribution of the unobserved heterogeneity $V$ is necessary. We make the following assumption about the distribution of $V$:

\begin{assumption}[Distribution of Unobserved Heterogeneity] \label{assum:distribution}
  The unobserved heterogeneity $V$ follows a known distribution with differentiable cumulative distribution function $F_V$, and $F_V$ is increasing in each dimension $j\in J$.
\end{assumption}

Assumption \ref{assum:distribution} requires that $V$ is dense everywhere in probability measure. 
I write the probability density function for $V$ as $f_V$. 
Armed with knowledge about the distribution of $V$, we can pinpoint the threshold function through the conditional probability distribution of treatment assignment. For a fixed treatment $k$, $\Pr(D=k|Z)$ can be directly observed from the sample and is treated as a function of $Z$. The conditional independence of $V$ and $Z$ in Assumption \ref{assum:independence} yields
\begin{eqnarray*}
  \Pr(D=k|Z) &=& \Pr( d_k(V,Q(Z))=1 |Z)  \\
   &=& \int  d_k(v,Q(Z)) f_V(v) dv
\end{eqnarray*}
where the fact $\mathds{1}\{ d_k(v,Q(Z))=1\} = d_k(v,Q(Z))$ is applied in the above equation. 
By inserting the treatment assignment decomposition from Equation \eqref{decomposedk} into the probability distribution, we obtain:
\begin{eqnarray}
  \Pr(D=k|Z)  
  &=& \int    \sum_{l\in \mathcal{L}} c_l^k \prod_{j\in l}  S_j(v,Q(Z)) f_V(v) dv \nonumber \\
  &=& \sum_{l\in \mathcal{L}} c_l^k \int_0^{\alpha_{l1}(Z)} \cdots \int_{0}^{\alpha_{lJ}(Z)} f_V(v_1,...,v_J) dv_1 \cdots dv_J   \nonumber \\
   &=& \sum_{l\in \mathcal{L}} c_l^k   F_V\left(\alpha_{l1}(Z),...,\alpha_{lJ}(Z)\right)    \label{idQ:probexpress}
\end{eqnarray}

Here, each $\alpha_{lj}(Z)$, $j=1,...,J$ is a function of $Z$ defined as:
\begin{equation*}
\alpha_{lj}(Z) = 
\begin{cases}
1, & \text{if } j \notin l \\
Q_j(Z), & \text{if } j \in l
\end{cases}
\end{equation*}

Combining Equation \eqref{idQ:probexpress} for $k=1,...,T$ yields a system of equations: 
\begin{eqnarray}
\sum_{l\in \mathcal{L}} c_l^1   F_V\left(\alpha_{l1}(Z),...,\alpha_{lJ}(Z)\right) &=& \Pr(D=1|Z)  \nonumber \\
 &\vdots &   \nonumber \\
\sum_{l\in \mathcal{L}} c_l^T   F_V\left(\alpha_{l1}(Z),...,\alpha_{lJ}(Z)\right) &=& \Pr(D=T|Z)  \label{idenQ:system}
\end{eqnarray}

Denote the matrix of coefficients in system of Equations \ref{idenQ:system} as $\{ c_l^k \}$. 
The following theorem states the identification of threshold function:

\begin{theorem}\label{thm:idenofQ}
Under Assumptions \ref{assum:independence} - \ref{assumption:threshold_involvement} and \ref{assum:distribution}, the system of Equations \eqref{idenQ:system} point identifies the threshold function $\{Q_j(Z)\}_{j=1}^J$ when 
$J \leq \text{rank} \{ c_l^k \} $.
\end{theorem}

In most cases, the only constraint on the system arises from the completeness condition assumed in Assumption \ref{assum:complete}, which implies that the sum of the probabilities for all treatments equals one. However, in some cases, additional constraints may be imposed on the probability distribution of each assigned treatment due to practical considerations of the treatment process. For example, if the treatment probabilities are restricted by external factors such as policy guidelines or operational limitations, this could further reduce the rank of the matrix.

However, 
obtaining an explicit solution for $\{Q_j(Z)\}_{j=1}^J$ from this equations system through mathematical manipulations like differentiation or matrix operations is challenging, even though the form of $F_V$ is known. This is primarily due to the fact that Equation \eqref{idQ:probexpress} does not exhibit a separable form for $Q_j(Z)$ without assuming specific forms for the cumulative distribution function $F_V$. Nevertheless, it is possible to gain insights into the threshold function using a numerical approach.


Before delving into the specifics of the numerical methods employed, we first establish the loss function to be used throughout our numerical experiments. Let us consider a proposed threshold function $\bar{Q}(\cdot)=\left( \bar{Q}_1(\cdot),...,\bar{Q}_J(\cdot) \right)$, under which we can define a corresponding loss function. 

The loss function measures the discrepancy between our model's predictions and the actual observed data. For this purpose, we use the sum of the squared differences between the computed probabilities under the proposed threshold function and the observed probabilities. More specifically, the loss function for a single observation $Z$ is given by
\begin{equation}\label{eq:lossfunction1}
    loss(Z;\bar{Q})= \sum_{k=1}^T \left( 
    \sum_{l\in \mathcal{L}} c_l^k   F_V\left(\bar{\alpha}_{l1}(Z),...,\bar{\alpha}_{lJ}(Z)\right) - \Pr(D=k|Z)
    \right)^2
\end{equation}

Here, $\bar{\alpha}_{lj}(Z)$ is defined as:
\begin{equation*}
\bar{\alpha}_{lj}(Z) = 
\begin{cases}
1, & \text{if } j \notin l \\
\bar{Q}_j(Z), & \text{if } j \in l
\end{cases}
\end{equation*}

The loss function for a given $\bar{Q}$, across all observations of instruments $Z$, can thus be written as the sum of the individual losses:
\begin{equation}\label{eq:lossfunc2}
    Loss(\bar{Q})=\sum_{o=1}^{N_o}   loss(Z^o;\bar{Q})
\end{equation}

This loss function, capturing the aggregate discrepancy between our model's predictions under the proposed threshold function and the observed data, will guide our subsequent numerical analyses.


To begin, we propose a parametric approximation of the true threshold function $Q$. This is realized by assuming a parametric form for $\bar{Q}(Z;\beta^Q)$, where $\beta^Q$ denotes the parameters. Formally, we have $\bar{Q}(Z;\beta^Q) = \left(\bar{Q}_1(Z;\beta^Q) ,..., \bar{Q}_J(Z;\beta^Q)\right)$. 

The approximation is performed using a set of basis functions $\{q_{jt_q}(\cdot)\}_{t_q=1}^{T_q}$, for each $j=1,...,J$. Here, $T_q \in \mathds{N}^+$ is the number of basis functions used for approximating each $\bar{Q}_j$. For the parameters, we let $\beta^Q_j=(\beta^Q_{j1},...,\beta^Q_{jT_q})\in \mathds{R}^{T_q}$, and $\beta^Q=(\beta^Q_1,...,\beta^Q_J)$ encompass the coefficients of the entire parametrization. 
Therefore, each $\bar{Q}_j$ is expressed as a linear combination of the basis functions, parameterized by $\beta^Q_{j}$, and can be written as follows:
\begin{equation*}
    \bar{Q}_j(\cdot;\beta^Q)=\sum_{t_q=1}^{T_q} \beta^Q_{jt_q} q_{jt_q}(\cdot)
\end{equation*}

By specifying $\bar{Q}$ in this way, we re-interpret the loss function as a function of the parameters $\beta^Q$. Consequently, finding an approximation to $Q(\cdot)$ is now equivalent to solving the following optimization problem:
\begin{equation*}
    \min_{\beta^Q} Loss( \bar{Q}(\cdot;\beta^Q) )
\end{equation*}

In solving the optimization problem, we can use algorithms like Gradient Descent or Newton's Method. This optimization yields an optimal set of parameters, denoted by $\hat{\beta}^Q$, which provides $\bar{Q}(\cdot;\hat{\beta}^Q)$ as an approximation to the true threshold function $Q(\cdot)$. The following theorem states the convergence of this approximation:

\begin{theorem}\label{thm:convergethreshold}
    Under Assumptions \ref{assum:independence} - \ref{assumption:threshold_involvement} and \ref{assum:distribution}, 
    if the parametric form $\bar{Q}(\cdot;\beta^Q)$ is correctly specified and $J \leq \text{rank} \{ c_l^k \}$, the estimator $\bar{Q}(\cdot;\hat{\beta}^Q)$ converges in probability to the true threshold function $Q(\cdot)$ as $N_0 \rightarrow \infty$.
\end{theorem}

Once the approximation is obtained, it will be used for further analysis in the subsequent sections of this paper.

\section{Identification of Marginal Treatment Response}\label{sec:identificationMTR}

In this section, I focus on identifying the marginal treatment response $E[Y_k|V=v]$. Since the identification of $Q(Z)$ has been addressed in Section \ref{ch:Q}, I will treat the threshold $Q(Z)$ as given and condition on $Q(Z)$ rather than $Z$ in the analysis.

Fixing a treatment $k$, we can observe $E[Y D_k | Q(Z)=q]$ from the sample data, where $q=(q_1,...,q_J) \in [0,1]^J$ is a realization of $Q(Z)$. By the model's construction, $Y D_k=Y_k D_k$. Thus, we can express:
\begin{align*}
  E[Y D_k | Q(Z)=q] &= E[Y_k D_k | Q(Z)=q] \\
  &= E[Y_k d_k(V, Q(Z)) | Q(Z)=q] \\
  &= E\left[ \left. E[Y_k d_k(V, Q(Z))|V, Q(Z)=q] \right| Q(Z)=q \right],
\end{align*}
where the last equality follows from the Law of Iterated Expectation.

Using the structure of conditional expectation and the independence of $Y_k$ and $V$ from $Z$, we have:
\begin{equation*}
    E[Y_k d_k(V, Q(Z))|V, Q(Z)=q] = E[Y_k | V] E[d_k(V, Q(Z))|V, Q(Z)=q].
\end{equation*}
Since $E[Y_k | V, Q(Z)=q] = E[Y_k | V]$ and $E[d_k(V, Q(Z))|V, Q(Z)=q] = d_k(V, q)$, it follows that:
\begin{align*}
    E[Y D_k | Q(Z)=q] &= E\left[ E[Y_k | V] d_k(V, q) \right| Q(Z)=q ] \\
    &= E\left[ E[Y_k | V] d_k(V, q) \right] \\
    &= \int d_k(v, q) E[Y_k | V=v] f_V(v) dv.
\end{align*}

Using the decomposition of $d_k(V, Q(Z))$ from Equation \eqref{decomposedk}, we can rewrite $d_k(V, q)$ as a sum of terms $l^1,...,l^N$, where $N \in \mathds{N}^+$ is the number of non-zero terms. Denote $l(q)$ as the set $\{ V: V_j < q_j, \forall j \in l \}$. We then have:
\begin{align}
    E[Y D_k | Q(Z)=q] &= \int \sum_{l \in \mathcal{L}} c_l^k \prod_{j \in l} S_j(v, q) E[Y_k | V=v] f_V(v) dv \notag \\
    &= \sum_{n=1}^N c_l^k \int_{l^n(q)} E[Y_k | V=v] f_V(v) dv. \label{mainleading}
\end{align}

Equation \eqref{mainleading} shows that the expected value $E[Y D_k | Q(Z)=q]$ is a sum of the marginal treatment responses integrated over the regions defined by $d_k(V, q)$. 

Unlike Equation \eqref{decompprob}, where $f_V(v)$ is independent of the treatment $k$, here $E[Y_k | V=v]$ varies with each treatment $k$. This adds complexity to the identification process because we need to account for the treatment-specific response.

To identify $E[Y_k | V=v]$ from observations, it is crucial to determine whether the leading term of $d_k(V, Q(Z))$ is full rank. This distinction is important and different from the discussion in Sections \ref{sec:fullrankleadingfv} and \ref{sec:nofullrankleadingfv}, where the focus was on the presence of at least one treatment with a full rank leading term. Here, we must consider each treatment $k$ separately. If a treatment does not have a full rank leading term, it falls into the category of no full rank leading term.

In this analysis, the variation of $E[Y_k | V=v]$ with different treatments means that we cannot combine equations across different treatments to increase the information available for identification, unlike in Section \ref{sec:nofullrankleadingfv}.

\subsection{Full Rank Leading Term}\label{sec:mtr:full}

When the treatment $k$ has a full rank leading term, we can follow the method described by \cite{lee2018identifying} to identify the marginal treatment response within our framework. This approach, similar to the one discussed in Section \ref{sec:fullrankleadingfv}, relies on the technique of taking derivatives to isolate $E[Y_k|V=v]$ from the summation and integration in Equation \eqref{mainleading}.

In Equation \eqref{mainleading}, the parameter $q$ determines the integration region through the set $l^n(q)$, for $n=1,...,N$. These regions expand or contract linearly with changes in $q$. Given that $f_V(v)$ is continuous by Assumptions \ref{continuousfV} and \ref{assum:distribution}, and $E[Y_k|V=v]$, for all $k=1,...,T$, is locally equicontinuous by Assumption \ref{assum:continuousmtr}, we can ensure the differentiability of Equation \eqref{mainleading} across all dimensions of $q$. Assumption \ref{interiorofQ} further guarantees that for any $q \in (0,1)^J$, Equation \eqref{mainleading} is well-defined within an open neighborhood of $q$. Therefore, we can differentiate Equation \eqref{mainleading} with respect to all components of $q$.\footnote{For detailed proof, see Section 6 of \cite{lee2018identifying}.}

Consider the full rank leading term of treatment $k$, denoted by $\tilde{l}$. When we differentiate $E[Y D_k | Q(Z)=q]$ with respect to all dimensions of $q=(q_1,...,q_J)$, all terms except the one involving $\tilde{l}$ become zero because they do not involve all dimensions of $q$. Thus, we have:
\begin{align*}
  \frac{\partial^J}{\partial q} E[Y D_k | Q(Z)=q] &= \frac{\partial^J}{\partial q} \left( c^k_{\tilde{l}} \int_{\tilde{l}(q)} E[Y_k | V=v] f_V(v) dv \right) \\
  &= \frac{\partial^J}{\partial q_1 \cdots \partial q_J} \int_0^{q_1} \cdots \int_0^{q_J} c_{\tilde{l}}^k E[Y_k | V=v] f_V(v) dv_1 \cdots dv_J \\
  &= c_{\tilde{l}}^k E[Y_k | V=q] f_V(q).
\end{align*}

From this derivative, we can identify the marginal treatment response as
\begin{equation}\label{eq:pointidentifiedMTRleesalanie}
  E[Y_k | V=v] = \frac{1}{c_{\tilde{l}}^k f_V(v)} \left. \frac{\partial^J \Pr(Y D_k = Y_k | Q(Z) = q)}{\partial q} \right|_{q=v} 
\end{equation}
almost everywhere over the entire support $v \in (0,1)^J$.






\subsection{Not Full Rank Leading Term}\label{sec:notfullrankleadingterm}

When the treatment $k$ does not have a full rank leading term, identifying the marginal treatment response becomes challenging. Some dimensions of $V$ are not involved in the treatment assignment process, which prevents the direct identification of treatment effects. However, by connecting these dimensions with observations from other treatments, we can still achieve identification, although it will be set identification rather than point identification.

Following the structure of analysis in Section \ref{sec:mtr:full}, we can differentiate the dimensions of $v$ associated with each leading term of treatment $k$. This will generate a conditional marginal treatment response for each leading term.

Let the leading terms of treatment $k$ be $l^1, ..., l^{N_k}$, where $N_k \in \mathds{N}^+$ is the number of leading terms for treatment $k$. Fix a leading term $l^i$. Recall the notation introduced in Section \ref{sec:nofullrankleadingfv}. By Equation \eqref{mainleading}, differentiating $E[Y D_k | Q(Z)=q]$ with respect to $q_{I^+_{l^i}}$ yields:
\begin{equation}\label{mtr:partialderive}
  \frac{\partial^{|l^i|} E[Y D_k | Q(Z)=q]}{\partial q_{I^+_{l^i}}} = 
  c^k_{l^i} \int 
      E[Y_k | V_{I^+_{l^i}} = q_{I^+_{l^i}}, V_{I^-_{l^i}} = v_{I^-_{l^i}} ] f_V(q_{I^+_{l^i}}, v_{I^-_{l^i}}) dv_{I^-_{l^i}},
\end{equation}
where Proposition \ref{prop3} ensures that differentiating all elements in a leading term eliminates other terms. The conditional expectation of $Y_k$ given $V$ in the dimensions involved in leading term $l$ is called the conditional marginal treatment response:
\begin{equation}\label{cmtr}
    E[Y_k | V_{I^+_{l^i}} = v_{I^+_{l^i}}]
  \equiv \int E[Y_k | V=v] f_{I^-_{l^i} | I^+_{l^i}}(v_{I^-_{l^i}} | v_{I^+_{l^i}}) dv_{I^-_{l^i}},
\end{equation} 
where $f_{I^-_{l^i} | I^+_{l^i}}(v_{I^-_{l^i}} | v_{I^+_{l^i}})$ is the conditional probability distribution function of $V_{I^-_{l^i}}$ given $V_{I^+_{l^i}}$, defined as:
\begin{equation*}
    f_{I^-_{l^i} | I^+_{l^i}}(v_{I^-_{l^i}} | v_{I^+_{l^i}}) \equiv \frac{f_V(v)}{f_{I^+_{l^i}}(v_{I^+_{l^i}})},
\end{equation*}
where $f_{I^+_{l^i}}$ is the marginal distribution of $V_{I^+_{l^i}}$ defined in Equation \eqref{marginalPDF}. Therefore, Equation \eqref{mtr:partialderive} implies:
\begin{equation}\label{eq:point.iden.of.deriv}
  \frac{\partial^{|l^i|} E[Y D_k | Q(Z)=q]}{\partial q_{I^+_{l^i}}} = 
  c^k_{l^i} E[Y_k | V_{I^+_{l^i}} = q_{I^+_{l^i}}] f_{I^+_{l^i}}(q_{I^+_{l^i}}),
\end{equation}
which allows us to point identify the conditional marginal treatment response $E[Y_k | V_{I^+_{l^i}} = v_{I^+_{l^i}}]$ almost everywhere over its support. 

Equation \eqref{cmtr} provides a decomposition of the conditional marginal treatment response by integrating out the dimensions not in the leading term, $I^-_{l^i}$. Combining all the leading terms of treatment $k$, we obtain a system of equations:
\begin{eqnarray}
  E[Y_k | V_{I^+_{i}} = v_{I^+_{i}}]
  &=& \int E[Y_k | V=v] f_{I^-_{i} | I^+_{i}}(v_{I^-_{i}} | v_{I^+_{i}}) dv_{I^-_{i}},   
  \label{mtr:eqsystem}
\end{eqnarray}
almost everywhere for all $v_{I^+_{l^i}} \in [0,1]^{|l^i|}$, $i=1, ..., N_k$.




However, identifying $E[Y_k|V=v]$ from this equation system is challenging, especially when $\bigcup\limits_{i=1}^{N_k} I^+_{l^i} \subsetneqq \mathbf{J}$. This means some dimensions $j\in \mathbf{J}$ are never involved in the treatment determination mechanism for treatment $k$. Nevertheless, by assuming the ranking of the marginal treatment response with respect to the level of treatment, we can achieve set identification for $E[Y_k|V=v]$. This assumption is stated as follows:

\begin{assumption}[Ranked Treatment]\label{assum:ranking}
  For any two treatments $\kappa, \kappa' \in \{1,...,T\}$, $E[Y_{\kappa'}|V] \leq E[Y_{\kappa}|V]$ if $\kappa' < \kappa$.
\end{assumption}

This assumption is similar to the monotone treatment response assumption in \cite{manski1997monotone}. 
Intuitively, Assumption \ref{assum:ranking} means that, conditional on the unobserved heterogeneity, higher treatment levels are expected to generate better outcomes. 
This assumption is particularly plausible in applications where the treatment variable represents an ordered intensity. For example, in education studies, treatment levels might correspond to the number of years of schooling completed; in labor economics, they might capture the duration or intensity of a training program; in health economics, they may represent different dosage levels of a medical intervention. In such cases, it is often reasonable to expect that higher treatment intensities yield systematically stronger effects on outcomes on average, conditional on unobserved characteristics. 
\footnote{This ranked treatment assumption can be extended to accommodate other plausible empirical settings. For instance, if higher treatment intensities are expected to produce lower outcomes on average, the assumption remains compatible with the framework after reversing the direction of the inequality. Moreover, even in cases where treatments are not ordinal in a global sense, such as when individuals are assigned to different physical training programs based on multidimensional characteristics like height, endurance, or strength, it may still be reasonable to assume a personalized ranking of treatment effects. That is, for each individual, conditional on $X$ and $V$, there exists a known ordering over treatments based on expected effectiveness. As long as such an individual-specific ordering can be specified over all treatments, the analytical structure developed in this paper remains applicable, with minor adjustments to the integration domain or inequality direction. Hence, the assumption is stated in a common form with minimized loss of generality.}

Note that the ranking of the marginal treatment response here is different from the monotonicity of treatments in previous econometric studies on treatment effects, such as \cite{imbens1994identification} and \cite{heckman2018unordered}. In their work, monotonicity indicates that when the instrument changes, there cannot be two individuals moving their treatment assignments in opposite directions. In other words, if the instrument causes one individual to switch from treatment group $k$ to treatment group $k'$, it should not cause another individual to switch from treatment group $k'$ to treatment group $k$. In this paper, the ranking of treatments assumes an order of marginal treatment response across different treatment levels. Monotonicity is not required in my setting, as the flexibility and high-dimensionality of the hyper-rectangle setup naturally allow for individuals to move in opposite directions.


Equipped with Assumption \ref{assum:ranking}, consider a treatment \(k' < k\), whose leading terms are denoted by \(l'^1, \ldots, l'^{N_{k'}}\), where \(N_{k'} \in \mathds{N}^+\) is the number of leading terms associated with treatment \(k'\). For any given leading term \(l'^{i'}\), Equation \eqref{cmtr} implies that
\begin{equation}\label{eq:ineq1}
    E[Y_{k'}|V_{I^+_{l'^{i'}}}=v_{I^+_{l'^{i'}}}] 
    \leq  \int E[Y_k|V=v] f_{I^-_{l'^{i'}}|I^+_{l'^{i'}}}(v_{I^-_{l'^{i'}}}|v_{I^+_{l'^{i'}}})  d v_{I^-_{l'^{i'}}},
\end{equation}
for almost every \(v_{I^+_{l'^{i'}}} \in [0,1]^{|l'^{i'}|}\), where \(i'=1,\ldots,N_{k'}\). The term \(E[Y_{k'}|V_{I^+_{l'^i}}=v_{I^+_{l'^i}}]\) is point identified from the data, analogous to Equation \eqref{eq:point.iden.of.deriv}. 

Similarly, for treatments satisfying \(k' > k\), an analogous inequality in the opposite direction holds. For each leading term \(l'^{i'}\) of treatment \(k'\), we have
\begin{equation}\label{eq:ineq2}
    E[Y_{k'}|V_{I^+_{l'^{i'}}}=v_{I^+_{l'^{i'}}}] 
    \geq  \int E[Y_k|V=v] f_{I^-_{l'^{i'}}|I^+_{l'^{i'}}}(v_{I^-_{l'^{i'}}}|v_{I^+_{l'^{i'}}})  d v_{I^-_{l'^{i'}}},
\end{equation}
for almost every \(v_{I^+_{l'^{i'}}} \in [0,1]^{|l'^{i'}|}\), with \(i' = 1, \ldots, N_{k'}\). 
For ease of exposition, I refer to the relationship in Equation \eqref{eq:ineq1} or \eqref{eq:ineq2} as the inequality constraint between treatment \(k\) and the leading term \(l'^{i'}\) of treatment \(k'\).

Given these constraints, we proceed to analyze the partial identification of marginal treatment responses. I first consider the identification of a single marginal treatment response and then extend the analysis to jointly characterize all marginal treatment responses across treatments.

\subsubsection{Identification of a Single Marginal Treatment Response}

Consider the identified set of a single marginal treatment response $E[Y_k|V=v]$. 
It is notable that the identification problem is no longer focusing on a single point $V=q$ but on a function over the interval $[0,1]^J$. The identified set for $E[Y_k|V=v]$ is then given by
\begin{eqnarray}
    \mathcal{I}^0_k = \{ E[Y_k|V=v]
   &:& \text{ Equation System \eqref{mtr:eqsystem} holds for all treatments } 1,...,T, \nonumber \\ 
   & & \text{ Equation System \eqref{eq:ineq1} holds for all treatments } k'<k, \nonumber \\
   & & \text{ Equation System \eqref{eq:ineq2} holds for all treatments } k<k'  \}   \label{identifiedset}
\end{eqnarray}
which is composed of a continuum of constraints for measurable, locally bounded, and locally equicontinuous function $E[Y_k|V=v]$ on $[0,1]^J$. 
The following theorem states that $\mathcal{I}^0_k$ is sharp.




\begin{theorem}[Sharpness of the Identified Set]\label{thm:sharpness}
Under Assumptions \ref{assum:independence}-\ref{assumption:threshold_involvement} and \ref{assum:ranking},  $\mathcal I_k^0$ is sharp for the marginal treatment response $E[Y_k\mid V]$.
\end{theorem}

Equation \eqref{identifiedset} defines a total of $\sum_{k=1}^T N_k$ constraints for the function $E[Y_k|V=v]$. However, not all of these constraints are active. We can simplify the identified set by eliminating slack constraints.

Specifically, consider the following two types of redundancy: 
First, suppose treatment $k$ has a leading term $l$, and treatment $k'$ has a leading term $l'$, with $l' \subseteq l$. In this case, the inequality constraint between $k$ and $l'$ is redundant. This is because the marginal treatment response can be point-identified conditional on $V_{I^+_l}$; through integration, one can also identify it conditional on $V_{I^+_{l'}}$, rendering the constraint involving $l'$ non-informative.

Second, consider three treatments $k'' < k' < k$, where $k''$ has a leading term $l''$ and $k'$ has a leading term $l'$ such that $l'' \subseteq l'$. I claim that the inequality constraint between $k''$ and $k$ via $l''$ can be eliminated. This follows from combining the inequality between $k''$ and $k'$ via $l''$ and that between $k'$ and $k$ via $l'$. Specifically:
\begin{eqnarray*}
    & & E[Y_{k''}|V_{I^+_{l''}}=v_{I^+_{l''}}] \\
    &\leq& \int E[Y_{k'}|V=v] f_{I^-_{l''}|I^+_{l''}}(v_{I^-_{l''}}|v_{I^+_{l''}})  d v_{I^-_{l''}} \\
    &=& \int \int E[Y_{k'}|V=v] 
    f_{I^-_{l'}|I^+_{l'}}(v_{I^-_{l'}}|v_{I^+_{l'}})
    f_{I^-_{l''}\backslash I^-_{l'}|I^-_{l'},I^+_{l''}} 
    (v_{I^-_{l''}\backslash I^-_{l'}|v_{I^-_{l'}},v_{I^+_{l''}}})
    d v_{I^-_{l''}\backslash I^-_{l'}}
    d v_{I^-_{l'}}  \\
    &=& \int E[Y_{k'}|V_{I^+_{l'}}=v_{I^+_{l'}}]
    f_{I^-_{l''}\backslash I^-_{l'}|I^-_{l'},I^+_{l''}} 
    (v_{I^-_{l''}\backslash I^-_{l'}}|v_{I^-_{l'}},v_{I^+_{l''}})
    d v_{I^-_{l''}\backslash I^-_{l'}} \\
    &\leq& \int \int E[Y_k|V=v] 
    f_{I^-_{l'}|I^+_{l'}}(v_{I^-_{l'}}|v_{I^+_{l'}})  
    f_{I^-_{l''}\backslash I^-_{l'}|I^-_{l'},I^+_{l''}} 
    (v_{I^-_{l''}\backslash I^-_{l'}}|v_{I^-_{l'}},v_{I^+_{l''}})
    d v_{I^-_{l'}}
    d v_{I^-_{l''}\backslash I^-_{l'}} \\
    &=& \int E[Y_{k}|V=v] f_{I^-_{l''}|I^+_{l''}}(v_{I^-_{l''}}|v_{I^+_{l''}})  d v_{I^-_{l''}} 
\end{eqnarray*}
Thus, the constraint between $k''$ and $k$ via $l''$ is implied by other constraints and can be removed. A similar argument applies when $k < k' < k''$ and $l'' \subseteq l'$. Note, however, that this logic does not extend to the cases where $l' \subseteq l''$, or where $k''<k<k'$.

Finally, after eliminating all such redundant inequalities, any remaining equation in \eqref{mtr:eqsystem} that involves only the leading terms removed in the above step can also be discarded, as they provide no additional identifying information.

Let $\mathcal{I}_k^1$ denote the simplified constraint set after this elimination. From the preceding analysis, it follows that $\mathcal{I}_k^1 = \mathcal{I}_k^0$.

\subsubsection{Joint Identification of Multiple Marginal Treatment Responses}


If we are going to jointly identify a set of marginal treatment responses $\{ E[Y_k|V=v] \}_{k\in K}$, where $K$ is a nonempty subset of $\{1,...,T\}$. 
Note that the joint identification will include internal/incur constraint between different marginal treatment responses, and we need to consider this constraint (interdependence introduced by pooling the marginal treatment responses together). 
Thus the identified set is given by
\begin{eqnarray}
    \mathcal I_K^0 = \{ 
    \{E[Y_k|V=v]\}_{k\in K}
    &:& 
    E[Y_k|V=v]\in \mathcal I_k^1, \ k\in K  \nonumber \\
    & &  E[Y_k|V=v] \leq E[Y_{k'}|V=v] \text{ if } k<k', \ k,k'\in K  \} \label{eq:jointidentset}
\end{eqnarray}
which satisfies both the constraint given by $\mathcal I_K^0$ in Equation \eqref{identifiedset} as well as the inter-treatment ranking constraints implied by Assumption \ref{assum:ranking}.

\begin{theorem}\label{thm:sharpness_multi}
Under Assumptions \ref{assum:independence}-\ref{assumption:threshold_involvement} and \ref{assum:ranking}, 
$\mathcal I_{K}^0$ is sharp for the collection $\{E[Y_k \mid V]\}_{k\in K}$.
\end{theorem}

There are also redundant constraints in Equation \eqref{eq:jointidentset}. 
First, if $k,k'\in K$, the inequality constraints between treatment $k$ and leading terms of $k'$ are redundant, since they are directly implied by the equality constraint of leading terms of $k'$ and the ranking constraint, and the same for the constraints between treatment $k'$ and leading terms of $k$. 
Second, we denote $\overline k_K \equiv \min\{\kappa\in K:\kappa> k \}$ and $\underline k_K \equiv \max\{\kappa\in K:\kappa< k \}$, 
if $k\notin K$, the inequality constraints between any treatment $k'\in K\backslash\{\overline k_K, \underline k_K \}$ and leading terms of $k$ are redundant, since they are implied by the constraint between treatment $\overline k_K$ or $\underline k_K$ with leading terms of $k$ and the ranking constraint. 

Similarly, we denote the identified set of marginal treatment responses after eliminating the redundant constraints as $\mathcal I_K^1$, thus $\mathcal I_K^1=\mathcal I_K^0$.


\subsubsection{Implementation}\label{sec:implement}



To characterize the identified set for marginal treatment responses $\{E[Y_k|V=v]\}_{k\in K}$, 
where $K$ may be a singleton or a set with multiple elements, I adopt a sieve-based parametric decomposition of each function $E[Y_k|V=v]$ with respect to $v = (v_1,...,v_J)$. The sieve method transforms the original infinite-dimensional problem into a finite-dimensional one, reducing computational complexity and improving tractability. However, this approximation introduces bias due to the limited expressiveness of finite basis expansions, and may inflate variance if the number of basis functions is chosen inappropriately.\footnote{See \cite{chen2007handbook} for a comprehensive treatment of sieve estimation.} 
\begin{equation*}
  E[Y_k|V=v]= \sum_{t=1}^{T_k} \theta_{kt}b_{kt}(v)
\end{equation*}
where $\{b_{kt}(\cdot)\}_{t=1}^{T_k}$ are known basis functions from $[0,1]^J$ to $\mathds{R}$, $T_k\in \mathds{N}^+$ is the number of basis functions, and $\theta_k \equiv (\theta_{k1},...,\theta_{kT_k}) \in \mathds{R}^{T_k}$ parameterizes the function $E[Y_k|V=v]$ for each $k\in K$.

We then define the feasible set for $\theta_K \equiv \{\theta_k\}_{k\in K}$ as:
\begin{equation}  \label{feasibletheta}
  \Theta_K=\left\{ \theta_K :   \left\{
  \sum_{t=1}^{T_k} \theta_{kt}b_{kt}(v)   \right\}_{k\in K}
  \in \mathcal{I}^1_K  \right\}
\end{equation}

This parameterization allows us to impose the identification constraints in $\mathcal{I}^1_K$ directly on the coefficients $\theta_K$, converting the infinite-dimensional constraint system into a finite-dimensional one. To implement this, we sample a finite set of $v$ values from its distribution. These sampled values are used to approximate the constraints in the system, enabling estimation of bounds for $\theta_K$.

Computing the feasible set $\Theta_K$ may be computationally demanding due to the potentially large number of constraints. To improve numerical feasibility and avoid empty identified sets, researchers can introduce slack terms into both equality and inequality constraints, which allow for small tolerances. Monte Carlo simulation methods are particularly useful in this setting. By repeatedly sampling $v$ and solving the associated inequality systems, we can numerically characterize $\Theta_K$.

It is important to note that the use of sieve approximation may result in a loss of accuracy of the identified set $\mathcal{M}_K$ because the true function $E[Y_k|V=v]$ may not lie exactly within the sieve space. This approximation error diminishes as the sieve space becomes richer, but in finite samples, the resulting $\mathcal{M}_K$ only approximates the identified set. Nevertheless, this approach provides a tractable and informative description of the treatment effects of interest.

The final form of the feasible identified set for $\{E[Y_k|V=v]\}_{k\in K}$ is:
\begin{equation}\label{mtrfeasibleset}
  \mathcal{M}_K \equiv \left\{ \left\{ \sum_{t=1}^{T_k} \theta_{kt}b_{kt}(v) \right\} : \theta_K \in \Theta_K  \right\}
\end{equation}
Equation \eqref{mtrfeasibleset} thus describes a feasible approximation of the identified set for marginal treatment responses for treatments $k \in K$, incorporating both model-based restrictions and computational considerations.

\section{Various Treatment Effects}\label{sec:various}

Identifying the marginal treatment response is essential, as it provides the foundation for defining a wide of treatment effects. In my multi-valued treatment setting, when restricting attention to comparisons between two treatment levels, one can define treatment effects analogous to those in binary treatment models. As shown in \cite{mogstad2018using}, these binary-type treatment effects can be expressed in terms of the marginal treatment response. 
Beyond such pairwise comparisons, the multi-valued treatment structure allows for the definition of more general treatment effects involving multiple treatment levels. Examples include:
\begin{itemize}
  \item \textbf{Marginal Treatment Effect} between treatments $k_1$ and $k_2$ conditional on $V$: 
  \begin{equation*}
      E[Y_{k_2}-Y_{k_1}|V] = E[Y_{k_2}|V] - E[Y_{k_1}|V]
  \end{equation*}

  \item \textbf{Average Treatment Effect} between treatments $k_1$ and $k_2$: 
  \begin{equation*}
      E[Y_{k_2}-Y_{k_1}] = \int (E[Y_{k_2}|V=v] - E[Y_{k_1}|V=v]) f_V(v) dv
  \end{equation*}

  \item \textbf{Average Treatment Effect on the Treated} between treatment $k_1$ and $k_2$ for group $D = k_3$, where $k_3 \in \{1,\dots,T\}$:\footnote{By Assumption \ref{assum:independence}, $E[Y_{k_2}-Y_{k_1}|D=k_3] = E[Y_{k_2}-Y_{k_1}|D=k_3,Z]$ for any instrument $Z$.}
  \begin{equation*}
      E[Y_{k_2}-Y_{k_1}|D=k_3] =
      \int (E[Y_{k_2}|V=v] - E[Y_{k_1}|V=v])
      \frac{d_{k_3}(v,Q(Z)) f_V(v)}{\Pr(D=k_3|Q(Z))} dv
  \end{equation*}

  \item \textbf{Local Average Treatment Effect} when instruments shift from $Z$ to $Z'$, with treatment changing from $k_1$ to $k_2$: 
  \begin{equation*}
   \frac{\int (E[Y_{k_2}|V=v] - E[Y_{k_1}|V=v]) d_{k_2}(v,Q(Z')) d_{k_1}(v,Q(Z)) f_V(v) dv}
   {\int d_{k_2}(v,Q(Z')) d_{k_1}(v,Q(Z)) f_V(v) dv}
  \end{equation*}
  provided the denominator is nonzero.

  \item \textbf{Policy Relevant Treatment Effect (PRTE)}: Consider a policy that changes the assignment mechanism from $(\mathbf{d}, Q)$ to $(\mathbf{d}', Q')$, where $\mathbf{d} = (d_1,\dots,d_T)$, and $\mathbf{d}' = (d'_1,\dots,d'_T)$ denotes the post-policy mechanism similarly, with $d_k'$ defined analogously to $d_k$ as Equation \eqref{decomposedk}. The PRTE conditional on $V$ and $Z$ is
  \begin{equation*}
      \sum_{k=1}^T (d'_k(V,Q'(Z)) - d_k(V,Q(Z))) E[Y_k|V]
  \end{equation*}
  and the \textbf{Average Policy Relevant Treatment Effect (APRTE)} conditional on $Z$ is
  \begin{equation} \label{eq:APRTE}
      \sum_{k=1}^T \int (d'_k(v,Q'(Z)) - d_k(v,Q(Z))) E[Y_k|V=v] f_V(v) dv
  \end{equation}
\end{itemize}

These treatment effects share a common structure: they are linear functionals of the marginal treatment responses. Accordingly, their identification reduces to identifying $\{E[Y_k|V=v]\}_{k\in K}$ for a subset $K \subseteq \{1,\dots,T\}$. Letting $G(\{E[Y_k|V=v]\}_{k\in K})$ denote any such treatment effect, the identified set is
\begin{equation}\label{eq:variousTE}
    \left\{ G(\{E[Y_k|V=v]\}_{k\in K}) : \{E[Y_k|V=v]\}_{k\in K} \in \mathcal{I}_K^1 \right\}
\end{equation}
This set can be computed using the same procedures introduced in Section \ref{sec:implement}. 
Although these treatment effects may not be point identified, the identified set in Equation \eqref{eq:variousTE} provides informative bounds that can guide empirical evaluation and policy analysis.

\section{Test on Policy Relevant Treatment Effect}\label{sec:test}







Equation \eqref{eq:APRTE} defines the APRTE when a policy changes the treatment assignment mechanism from $(\mathbf{d}, Q)$ to $(\mathbf{d}', Q')$, conditional on $Z$. 
To evaluate the aggregate welfare impact of such policy changes, I define the Gross Policy Relevant Treatment Effect (GPRTE) based on $N_o$ observed units as:
\begin{equation} \label{eq:GPRTE}
    \sum_{o=1}^{N_o} \omega_o \sum_{k=1}^T \int \left(d'_k(v, Q'(Z^o)) - d_k(v, Q(Z^o))\right) E[Y_k | V = v] f_V(v) \, dv,
\end{equation}
where $\omega_o \geq 0$ denotes the weight assigned to observation $o$, reflecting policy makers' welfare objectives. Without loss of generality, I normalize $\sum_{o=1}^{N_o} \omega_o = 1$. A common choice is $\omega_o = 1/N_o$.

Such scenarios frequently arise in policymaking contexts. For example, in labor training programs as described in Section \ref{sec:intro}, when individuals are assigned to different training contents, a policy reform may replace any discretionary assignment with a talent-based mechanism that depends more systematically on individual-level data. 
Similarly, in education, students may receive varying levels of instructional support under practices, while a new policy might implement a standardized test-based allocation rule. In both examples, the policy shifts the assignment mechanism and the resulting effect on social outcomes is captured by the GPRTE.

A natural question arises as to whether the proposed policy improves social outcomes. To address this, I develop a framework to test whether the GPRTE is statistically different from zero. For notational convenience, define the individual-level contribution:
\begin{equation*}
   \Delta \mu^o (\mathbf{d}, Q, \mathbf{Z}^{N_o})
   \equiv  
   \sum_{k=1}^T \int \left(d'_k(v, Q'(Z^o)) - d_k(v, Q(Z^o))\right) E[Y_k | V = v] f_V(v) \, dv,
\end{equation*}
where $\mathbf{Z}^{N_o} \equiv (Z^1, \ldots, Z^{N_o})$. Then the GPRTE can be expressed as
\[
   \Delta W \equiv \sum_{o=1}^{N_o} \omega_o \Delta \mu^o (\mathbf{d}, Q, \mathbf{Z}^{N_o}),
\]
where the dependence of $\Delta W$ on $\mathbf{d}', Q', \mathbf{d}, Q$, and $\mathbf{Z}^{N_o}$ is suppressed for simplicity.

A key challenge is that the treatment effects $E[Y_k | V = v]$ may not be point identified, as discussed in Section~\ref{sec:identificationMTR}. Therefore, I consider two cases, when the GPRTE is point identified and when it is set identified, and develop appropriate testing methods for each.

\subsection{Test for Point Identified Treatment Effect}

I first consider the case in which the GPRTE is point identified. This occurs when all treatment levels involved in the policy change, both before and after implementation, have full-rank leading terms. Under this condition, the marginal treatment responses for the relevant treatments are point identified, and so is the GPRTE defined in Equation~\eqref{eq:GPRTE}.

To assess the impact of the policy, the hypothesis testing problem is specified as:
\begin{eqnarray*}
    H_0: & \quad & \Delta W = 0  \\
    H_1: & \quad & \Delta W \neq 0
\end{eqnarray*}
where $\Delta W$ denotes the GPRTE, representing the aggregate effect of the policy intervention on social outcomes.

The first step is to compute the observed GPRTE from the data. The following algorithm outlines this procedure:

\begin{algorithm}[Computation of Point Identified GPRTE]\label{alg:GPRTE}
Given observations $ (Y^o, D^o, Z^o, X^o)$ for $o = 1, ..., N_o$ and the conditional distribution of unobserved heterogeneity $f(V)$, proceed as follows:
\begin{enumerate}[label=Step \arabic*:, align=left]
  \item For each observation $o = 1, ..., N_o$, independently draw \(M\) samples from the distribution of \(V\) conditional on $Z^o$, denoted by \(v_1^o, v_2^o, ..., v_M^o\). For each draw \(v_m^o\), compute:
    \begin{itemize}
        \item The treatment assignment under the baseline and counterfactual policies:
        \[
        \bar D_m^o = d(v_m^o, Q(Z^o)), \quad 
        \bar D'^o_m = d'(v_m^o, Q'(Z^o)).
        \]
        \item The corresponding marginal treatment responses: \footnote{Note that these expressions depend on $Z^o$ because the instruments enter the treatment assignment rule. For example,
        \[
        E[Y_{\bar D_m^o} | V = v_m^o, Z^o] 
        = \sum_{k=1}^T d_k(v_m^o, Q(Z^o)) E[Y_k | V = v_m^o],
        \]
        and similarly for the counterfactual expression.}
        \[
        E[Y_{\bar D_m^o} | V = v_m^o, Z^o], \quad 
        E[Y_{\bar D'^o_m} | V = v_m^o, Z^o].
        \]
    \end{itemize}
    Average over $m = 1, ..., M$ to obtain the individual-level expected outcomes before and after the policy:
    \begin{equation*}
        \mu'^o_M = \frac{1}{M} \sum_{m=1}^M E[Y_{\bar D'^o_m} | V = v_m^o, Z^o], \quad
        \mu^o_M = \frac{1}{M} \sum_{m=1}^M E[Y_{\bar D_m^o} | V = v_m^o, Z^o].
    \end{equation*}
    Define the difference as:
    \[
    \Delta\mu^o_M \equiv \mu'^o_M - \mu^o_M.
    \]
  \item Aggregate across observations using policy weights $\omega_o$ to compute the observed GPRTE:
  \begin{equation}\label{eq:observedGPRTE}
      \Delta W_{N_o} = \sum_{o=1}^{N_o} \omega_o \, \Delta\mu^o_M,
  \end{equation}
  which serves as the point estimate of the gross policy-relevant treatment effect.
\end{enumerate}
\end{algorithm}




Consider the case where $N_o$ is fixed and $M \rightarrow \infty$. For a given observation $o$, and conditional on $Z^o$, the quantity 
\[
E[Y_{\bar{D}'^o_m} | V = v_m^o, Z^o] - E[Y_{\bar{D}_m^o} | V = v_m^o, Z^o]
=
\sum_{k=1}^T (d'_k(v_m^o,Q'(Z^o))-d_k(v_m^o,Q(Z^o))) E[Y_{k}|V=v_m^o]
\]
is i.i.d.\ across $m = 1, \dots, M$, with bounded mean and variance. This follows because the treatment assignment indicators $d_k(v_m^o, Q(Z^o))$ and $d'_k(v_m^o, Q'(Z^o))$ are binary and mutually exclusive across $k = 1, \dots, T$, and the marginal treatment responses $E[Y_k|V=v]$ are bounded. Hence, their difference is a bounded weighted sum and has finite conditional expectation and variance. 

Now, consider the observed GPRTE in Equation~\eqref{eq:observedGPRTE}. To proceed, we first verify that the conditional expectation of $\Delta W_{N_o}$ equals the true quantity $\Delta W$:

\begin{lemma}\label{lemma:unbiased_GPRTE}
The computed $\Delta W_{N_o}$ in Equation~\eqref{eq:observedGPRTE} is an unbiased estimator of the GPRTE defined in Equation~\eqref{eq:GPRTE}. Specifically,
\begin{equation*}
    \mathbb{E}[\Delta W_{N_o} \mid \mathbf{Z}] = \Delta W.
\end{equation*}
\end{lemma}

Given the unbiasedness of the observed GPRTE estimator, we can analyze the asymptotic distribution of the test statistic, and establish the hypothesis testing framework when \( M \to \infty \) and \( N_o \) is fixed. 
For a given observation \( o \), by Central Limit Theorem, the sample mean \( \Delta \mu^o_M \) satisfies:
\[
\Delta \mu^o_M \sim N\left(\mu'^o_{E[Y|V]} - \mu^o_{E[Y|V]}, \frac{\sigma_o^2}{M}\right),
\]
where \( \sigma_o^2 \) is the variance of the individual differences in marginal treatment responses. Since \( \Delta W_{N_o} \) is a weighted sum of \( \Delta \mu^o_M \), it follows that:
\begin{equation}\label{eq:DeltaWno}
    \Delta W_{N_o} \sim N\left(\Delta W, \sum_{o=1}^{N_o} \omega_o^2 \frac{\sigma_o^2}{M}\right).
\end{equation}
To estimate this variance in Equation \eqref{eq:DeltaWno}, we use the sample variance for each observation \( o \):
\[
\widehat{\sigma}_o^2 = \frac{1}{M-1} \sum_{m=1}^M \left( \Delta Y_m^o - \Delta \mu^o_M \right)^2,
\]
where \( \Delta Y_m^o = E[Y_{\bar{D}'^o_m} | V = v_m^o, Z^o] - E[Y_{\bar{D}_m^o} | V = v_m^o, Z^o] \). The estimated variance of \( \Delta W_{N_o} \) is then:
\[
\widehat{\mathrm{Var}}(\Delta W_{N_o}) = \sum_{o=1}^{N_o} \omega_o^2 \frac{\widehat{\sigma}_o^2}{M}.
\]

Under the null hypothesis \( H_0: \Delta W = 0 \), the test statistic is constructed as:
\[
Z_{\text{test}} = \frac{\Delta W_{N_o}}{\sqrt{\widehat{\mathrm{Var}}(\Delta W_{N_o})}},
\]
When \( M \to \infty \), the numerator asymptotically follows a normal distribution, and the denominator converges to a constant. By Slutsky's theorem, the test statistic \( Z_{\text{test}} \) asymptotically follows a standard normal distribution:
\[
Z_{\text{test}} \sim N(0, 1), \quad \text{under } H_0.
\]

With critical value from standard normal distribution, 
this test evaluates the statistical significance of the GPRTE under the asymptotic framework where \( M \to \infty \) and \( N_o \) is fixed.


\subsection{Test for Set Identified Treatment Effect}

Consider the case when the marginal treatment response is not point identified. In this situation, GPRTE is represented as a set rather than a single point estimate. This necessitates a different approach for hypothesis testing. 

Our null and alternative hypothesis: 
\begin{eqnarray*}
    H_0: & \quad & \Delta W = 0  \\
    H_1: & \quad & \Delta W \neq 0
\end{eqnarray*}

Let \(\Delta \mathcal{W}\) denote the set of possible values for the GPRTE under the policy change from \((\mathbf{d},Q)\) to \((\mathbf{d}',Q')\), we know from Equation \eqref{eq:variousTE},
\begin{equation*}
    \Delta \mathcal{W} = \{ \Delta W : \Delta W =  \sum_{o=1}^{N_o} \omega_o \Delta \mu^o (\mathbf{d},Q,\mathbf Z^{N_o}) , \ E[Y_k|V=v] \in \mathcal{I}_K^1 \}
\end{equation*}

The region \(\mathcal{I}_K^1\) for each marginal treatment response \(E[Y_k|V=v]\) is defined by a set of inequalities constraints, which determine a continuous set of possible values for \(E[Y_k|V=v]\). 
The GPRTE, \(\Delta W = \sum_{o=1}^{N_o} \omega_o \Delta \mu^o\), is a weighted sum of \(\Delta \mu^o\), where \(\Delta \mu^o \in \mathcal{I}_k^1\). Since weighted sums of continuous sets are also continuous, \(\Delta \mathcal{W}\) is a connected interval.
Inequalities defining \(\mathcal{I}_k^1\) hold with equality at the boundaries, \(\Delta \mathcal{W}\) is a closed. 
As a consequence, $\Delta\mathcal W$ is a closed interval. 
I use the following algorithm to compute it.

\begin{algorithm}[Computation of Set Identified GPRTE]\label{alg:setGPRTE}
Given the observation $ (Y^o, D^o, Z^o, X^o)$, $o = 1, ..., N_o $ in the dataset and the conditional distribution of unobserved heterogeneity $f(V)$, 
\begin{enumerate}[label=Step \arabic*:, align=left]
  \item For each observation $o=1,...,N_o$, draw \(M\) samples from the conditional distribution of \(V\), denoted as \(v_1^o, v_2^o, ..., v_M^o\). For each draw \(v_m^o\),
    \begin{itemize}
        \item Under each mechanism \( (\mathbf{d},Q) \) and \( (\mathbf{d}',Q') \), determine the treatment assignment $\bar D_m^o = d(v_m^o,Q(Z^o))$ and $\bar D'^o_m= d'(v_m^o,Q'(Z^o))$ respectively.
        \item Calculate the set identified  $E[Y_{\bar D'^o_m}|V=v_m^o, Z^o]-E[Y_{\bar D_m^o}|V=v_m^o, Z^o]$, and denote its lower and upper bounds as $\underline {\delta^o_m}$ and $\overline{ \delta^o_m}$. 
    \end{itemize}
  \item Average the lower and upper bounds respectively over $m=1,...,M$:
    \begin{equation*}
      \underline{\Delta \mu^o_M} = \frac{1}{ M} \sum_{m=1}^M \underline{ \delta^o_m}
     \, , \quad \ 
        \overline{\Delta \mu^o_M} = \frac{1}{ M} \sum_{m=1}^M \overline{ \delta^o_m}
    \end{equation*}
  \item Sum up all the observations \(o=1,...,N_o\) with weights, and calculate the lower and upper bounds of observed GPRTE
    \begin{equation*}  
   \underline{\Delta W_{N_o}} =  \sum_{o=1}^{N_o} \omega_o \left(  \underline{\Delta \mu^o_M}  \right)   
   \, , \quad
   \overline{\Delta W_{N_o}} =  \sum_{o=1}^{N_o} \omega_o \left(  \overline{\Delta \mu^o_M}  \right) 
\end{equation*}
The GPRTE is represented by $[\underline{\Delta W_{N_o}},\overline{\Delta W_{N_o}}]$.
\end{enumerate}
\end{algorithm}


For each observation \(o\), the lower bounds \(\underline{\delta^o_m}\), \(m = 1, \dots, M\), are i.i.d. because they are functions of independently drawn values \(v_1^o, \dots, v_M^o\). Similarly, the upper bounds \(\overline{\delta^o_m}\) are also i.i.d. across \(m\). However, within the same index \(m\), the pair \((\underline{\delta^o_m}, \overline{\delta^o_m})\) are not independent, as both depend on the same draw \(v_m^o\). Since the outcome variable is bounded, both lower and upper bounds are finite.

As a result, for a given \(o\), the quantities \(\underline{\Delta \mu^o_M}\) and \(\overline{\Delta \mu^o_M}\), which are empirical averages over \(M\) i.i.d. samples, satisfy the conditions of the Central Limit Theorem. Hence,  \(\sqrt{M} \, \underline{\Delta \mu^o_M}\) and \(\sqrt{M} \, \overline{\Delta \mu^o_M}\), each converge in distribution to a normal distribution. Consequently, the bounds of GPRTE \(\sqrt{M} \, (\underline{\Delta W_{N_o}}, \overline{\Delta W_{N_o}})\), follows a bivariate normal distribution, which we denote as
\begin{equation*}
  \sqrt M  \begin{pmatrix}
\underline{\Delta W_{N_o}} \\
\overline{\Delta W_{N_o}}
\end{pmatrix}
\sim \mathcal N\left(
\begin{pmatrix}
\mu_{\underline{\Delta W}} \\
\mu_{\overline{\Delta W}}
\end{pmatrix},
\begin{pmatrix}
\sigma^2_{\underline{\Delta W}} & 
\rho_{\underline{\Delta W},\overline{\Delta W}} \sigma_{\underline{\Delta W}} \sigma_{\overline{\Delta W}} \\
\rho_{\underline{\Delta W},\overline{\Delta W}} \sigma_{\underline{\Delta W}} \sigma_{\overline{\Delta W}} 
&  \sigma^2_{\overline{\Delta W}}
\end{pmatrix}
\right),
\end{equation*}


To construct the variance-covariance matrix of the joint normal distribution for \( \sqrt{M} (\underline{\Delta W_{N_o}}, \overline{\Delta W_{N_o}}) \), we estimate \( \sigma_{\underline{\Delta W}}, \sigma_{\overline{\Delta W}}, \) and \( \rho_{\underline{\Delta W}, \overline{\Delta W}} \). These components are derived from the variances and covariances of the lower and upper bounds across the observations and weights.

The variance of the lower bound is defined as:
\[
\sigma_{\underline{\Delta W}}^2 = \sum_{o=1}^{N_o} \omega_o^2 \frac{\sigma_{\underline{\Delta \mu^o}}^2}{M},
\]
where \( \sigma_{\underline{\Delta \mu^o}}^2 \) is the variance of the sample lower bounds for observation \( o \). The unbiased estimator is:
\[
\widehat{\sigma}_{\underline{\Delta \mu^o}}^2 = \frac{1}{M-1} \sum_{m=1}^M \left( \underline{\delta^o_m} - \underline{\Delta \mu^o_M} \right)^2,
\]
where \( \underline{\Delta \mu^o_M} = \frac{1}{M} \sum_{m=1}^M \underline{\delta^o_m} \) is the sample mean of the lower bounds for observation \( o \). Substituting this into the variance formula, the estimator for \( \sigma_{\underline{\Delta W}}^2 \) is:
\[
\widehat{\sigma}_{\underline{\Delta W}}^2 = \sum_{o=1}^{N_o} \omega_o^2 \frac{\widehat{\sigma}_{\underline{\Delta \mu^o}}^2}{M}.
\]

Similarly, the estimator for the variance of the upper bound $\sigma_{\overline{\Delta W}}^2$ is 
\[
\widehat{\sigma}_{\overline{\Delta W}}^2 = \sum_{o=1}^{N_o} \omega_o^2 \frac{\widehat{\sigma}_{\overline{\Delta \mu^o}}^2}{M}.
\]
where
\[
\widehat{\sigma}_{\overline{\Delta \mu^o}}^2 = \frac{1}{M-1} \sum_{m=1}^M \left( \overline{\delta^o_m} - \overline{\Delta \mu^o_M} \right)^2,
\]


Denote the difference between the upper and lower bounds as \(R_{N_o}\):
\[
R_{N_o} = \overline{\Delta W_{N_o}} - \underline{\Delta W_{N_o}}. 
\]
We construct a confidence interval as
\begin{equation*}
C^{\Delta W}_{1-\alpha} = \left[ \underline{\Delta W_{N_o}} - \overline{C}_M \widehat{\sigma}_{\underline{\Delta W}} / \sqrt{M}, \overline{\Delta W_{N_o}} + \overline{C}_M \widehat{\sigma}_{\overline{\Delta W}} / \sqrt{M} \right],
\end{equation*}
where \( \overline{C}_M \) satisfies
\begin{equation*}
\Phi \left( \overline{C}_M + \sqrt{M} \cdot \frac{ R_{N_o} }{\max(\widehat{\sigma}_{\underline{\Delta W}}, \widehat{\sigma}_{\overline{\Delta W}})} \right) - \Phi(\overline{C}_M) = 1- \alpha
\end{equation*}
and $\alpha$ is the chosen significance level. 
We state the following theorem:

\begin{theorem}\label{thm:coveragerate}
    $C^{\Delta W}_{1-\alpha}$ asymptotically cover $\Delta W$ by
\begin{equation}\label{eq:asympcover}
    \lim_{M \to \infty} \inf_{ E[Y_k|V=v] \in \mathcal{I}_K^1} \Pr(\Delta W \in C^{\Delta W}_{1-\alpha}) \geq 1-\alpha.
\end{equation}
\end{theorem}

Therefore, with the given significance level $\alpha$, we can perform the hypothesis test. If $0\notin C^{\Delta W}_{1-\alpha}$, we can reject the null hypothesis that GRPTE equals zero.


\section{Conclusion}\label{sec:conclusion}

This paper develops a hyper-rectangle model to address the challenges of analyzing treatment effects in micro-econometric studies with set-identified parameters. By introducing a framework that leverages the interplay among observed outcomes, treatments, covariates, unobserved heterogeneity, and instrumental variables, the proposed model enhances the identification and estimation of treatment effects under realistic and flexible assumptions. 

A key contribution of this study is the identification of the Marginal Treatment Response  function and the threshold function $Q(Z)$, which are fundamental for understanding heterogeneous treatment effects. By distinguishing between cases where the threshold function or the distribution of unobserved heterogeneity is known, the model provides a structured approach to derive bounds for treatment effects. This approach offers practical insights into policy-relevant questions by accommodating partial identification and allowing for robust inference.

The model's ability to handle complex treatment assignment mechanisms expands its applicability to a wide range of empirical contexts. It facilitates the estimation of various treatment effect measures, including Average Treatment Effects and Policy Relevant Treatment Effects, while remaining computationally tractable. Moreover, the framework provides tools for hypothesis testing, enabling researchers to draw meaningful conclusions even under partial identification.

Overall, the proposed methodology contributes to the econometric literature by offering a flexible and empirically applicable tool for treatment effect analysis. Future research could extend this framework by exploring dynamic treatment settings, incorporating additional sources of uncertainty, or applying the methodology to large-scale datasets in practice.

\clearpage

\bibliographystyle{ecta}
\bibliography{citation.bib}

\appendix

\section{Proofs}

\begin{proof}[Proof of Theorem \ref{thm:idenofQ}]
Consider the system of $T$ equations given in system of Equations \eqref{idenQ:system}, which consists of $J$ unknowns $\{Q_j(Z)\}_{j=1}^J$. Assumption \ref{assumption:threshold_involvement} implies each unknown threshold function $Q_j(Z)$, $j \in \{1,\ldots,J\}$ appears in the system through the functions $\alpha_{lj}(Z)$.

First, I argue that the vector $Q(Z) = \left(Q_1(Z),...,Q_J(Z)\right)$ has a one-to-one correspondence with the vector 
$\bar F(Z) \equiv \left(F_V\left(\alpha_{11}(Z),...,\alpha_{1J}(Z)\right),...,F_V\left(\alpha_{2^J-1,1}(Z),...,\alpha_{2^J-1,J}(Z)\right)\right)$. 
On one hand, from the construction of $\alpha_{lj}(Z)$, we know each $Q(Z)$ uniquely determines $\bar F(Z)$. 
On the other hand, for every $j_0\in\{1,...,J\}$, there is a subset $l^{j_0}= \{j_0\} \subset \mathcal L$ such that
$F_V\left(\alpha_{l^{j_0}1}(Z),...,\alpha_{l^{j_0}J}(Z)\right)$ 
has all inputs equal to 1 except for the $j_0$th element, which is equal to $Q_{j_0}(Z)$. By Assumption \ref{assum:distribution}, $F_V$ is strictly increasing in each argument, meaning that $Q_j(Z)$, $j=1,...,J$ can be uniquely recovered from the corresponding element of $\bar F(Z)$. 
Therefore, $\bar F(Z)$ and $Q(Z)$ are equivalent in terms of the unknowns they contain. 

For the system to have a unique solution, we require that the number of independent equations provided by the matrix $\{ c_l^k \}$ is at least $J$, the number of unknowns. 
Thus, if the rank condition $J \leq \text{rank} \{ c_l^k \}$ holds, the system can be solved uniquely for the threshold functions.
\end{proof}

\begin{proof}[Proof of Theorem \ref{thm:convergethreshold}]
    First, consider the system of Equations \eqref{idenQ:system} for a fixed value of $Z$. For the parametric form $\bar{Q}(Z; \beta^Q)$ with the true parameter vector $\beta^Q$, the corresponding loss function is defined in Equation \eqref{eq:lossfunction1}. As implied by Theorem \ref{thm:idenofQ}, the true threshold function $Q(Z)$ uniquely solves this system, ensuring that it minimizes the loss function. Therefore, if the parametric form $\bar{Q}(Z; \beta^Q)$ is correctly specified, the true parameter vector $\beta^Q$ is guaranteed to be the global minimizer of the loss function $loss(Z; \bar{Q})$.

    Next, consider the empirical average of the loss function over the sample in Equation \eqref{eq:lossfunc2}. As $N_0 \rightarrow \infty$, by the Weak Law of Large Numbers, the empirical average loss function converges in probability to the expected loss function:
    \begin{equation*}
        \lim_{N_0 \rightarrow \infty} \Pr\left(|Loss(\bar{Q}) - E[loss(Z;\bar{Q})]| > \epsilon\right) = 0
    \end{equation*}
    for any $\epsilon > 0$. 

    The expected loss function $E[loss(Z;\bar{Q})]$ is the expected value of the loss function over the distribution of $Z$. Since Theorem \ref{thm:idenofQ} guarantees that for each fixed $Z$, the true parameter $\beta^Q$ minimizes the loss function, the true parameter also minimizes the expected loss function. 
    Since the empirical loss function $Loss(\bar{Q})$ converges to the expected loss function, the estimator $\hat{\beta}^Q$ obtained by minimizing $Loss(\bar{Q})$ will converge to the true parameter $\beta^Q$ as $N_0 \rightarrow \infty$. 
    Therefore, $\hat{\beta}^Q$ is a consistent estimator for $\beta^Q$, and $\bar{Q}(\cdot;\hat{\beta}^Q)$ converges in probability to the true threshold function $Q(\cdot)$ as the sample size increases.
\end{proof}

\begin{lemma}\label{lem:lead-completeness}
Let $k\in\{1,\ldots,T\}$ and let $m_k:(0,1)^J\to(0,\bar Y)$ be admissible (measurable, locally bounded, locally equicontinuous) and satisfy that for every leading term $l$ of treatment $k$ and for a.e. $v_{I_l^+}\in(0,1)^{|I_l^+|}$,
\begin{equation}\label{eq:lead-equality}
\int m_k(v_{I_l^+},v_{I_l^-})\, f_{I_l^-\mid I_l^+}(v_{I_l^-}\mid v_{I_l^+})\,dv_{I_l^-}
\;=\;
m_{k,l}(v_{I_l^+}),
\end{equation}
where $m_{k,l}(\cdot)$ is the conditional marginal treatment response identified from the data via Equation \eqref{eq:point.iden.of.deriv}. Then, for a.e. $q\in(0,1)^J$,
\begin{equation}\label{eq:full-moment}
E[Y D_k\mid Q(Z)=q]\;=\;\int d_k(v,q)\, m_k(v)\, f_V(v)\, dv.
\end{equation}
\end{lemma}

\begin{proof}[Proof of Lemma \ref{lem:lead-completeness}]
Write $S_j(v,q)=\mathds{1}\{v_j<q_j\}$ and $d_k(v,q)=\sum_{l\in\mathcal L_k} c_l^k \prod_{j\in I_l^+} S_j(v,q)$. Define
\[
\Psi_k(q)\;\equiv\;\int d_k(v,q)\, m_k(v)\, f_V(v)\, dv,\qquad
G_k(q)\;\equiv\;E[Y D_k\mid Q(Z)=q].
\]
By bounded support of $Y$ and the regularity assumed for $f_V$ and $m_k$, both $\Psi_k$ and $G_k$ are Lebesgue a.e. differentiable in each $q_j$ on $(0,1)^J$ and absolutely continuous on rectangles.

Fix a leading term $l$ of treatment $k$. Differentiating $\Psi_k$ w.r.t. the coordinates in $I_l^+$ and using the Leibniz rule for integrals over variable limits yields, for a.e. $q_{I_l^+}$,
\begin{align*}
\frac{\partial^{|I_l^+|}}{\partial q_{I_l^+}} \Psi_k(q)
&= c_l^k \int m_k\big(q_{I_l^+},v_{I_l^-}\big)\, f_V\big(q_{I_l^+},v_{I_l^-}\big)\, dv_{I_l^-} \\
&= c_l^k\, m_{k,l}(q_{I_l^+})\, f_{I_l^+}(q_{I_l^+}),
\end{align*}
where the second equality follows from the definition of the conditional density and Equation \eqref{eq:lead-equality}. On the other hand, by  Equation \eqref{eq:point.iden.of.deriv},
\[
\frac{\partial^{|I_l^+|}}{\partial q_{I_l^+}} G_k(q)
\;=\; c_l^k\, m_{k,l}(q_{I_l^+})\, f_{I_l^+}(q_{I_l^+})
\quad \text{for a.e. } q_{I_l^+}.
\]
Thus $\frac{\partial^{|I_l^+|}}{\partial q_{I_l^+}} \Psi_k(q)=\frac{\partial^{|I_l^+|}}{\partial q_{I_l^+}} G_k(q)$ a.e. for every leading term $l$.

Finally, note that $\Psi_k(q)=G_k(q)=0$ whenever any coordinate of $q$ equals $0$ (the integration region collapses). By absolute continuity on rectangles and the fundamental theorem of calculus in several variables, equality of all these mixed partial derivatives together with the common boundary value implies $\Psi_k(q)=G_k(q)$ for a.e. $q\in(0,1)^J$. 
\end{proof}

\begin{lemma}\label{lem:feasible-completion}
Fix admissible functions $m_1,\ldots,m_T$ such that, for each treatment $k$:
\begin{itemize}
\item[(a)] $m_k$ satisfies the leading-term equalities \eqref{eq:lead-equality} for every leading term $l$ of treatment $k$;
\item[(b)] the ranking inequalities \eqref{eq:ineq1} and \eqref{eq:ineq2} hold a.e. 
\end{itemize}
Then there exists a probability space supporting latent $V$ with density $f_V$, instruments $Z$, and potential outcomes $(Y_1,\ldots,Y_T)$ such that:
\begin{enumerate}
\item $E[Y_k\mid V=v]=m_k(v)$ for a.e. $v$, and $0<Y_k<\bar Y$ a.s.;
\item $Y_k \perp Z \mid V$ for all $k$; 
\item for all $k$ and a.e. $q$, $E[Y D_k\mid Q(Z)=q]$ computed under $D=d(V,Q(Z))$ equals the observed value in the data.
\end{enumerate}
\end{lemma}

\begin{proof}[Proof of Lemma \ref{lem:feasible-completion}]
Let $U_1,\ldots,U_T$ be i.i.d. $\mathrm{Uniform}(0,1)$, independent of $(V,Z)$. For each $k$ and $v\in(0,1)^J$, choose any distribution $G_k(v,\cdot)$ on $(0,\bar Y)$ with mean $m_k(v)$ (e.g., a two-point distribution at $a(v),b(v)\in(0,\bar Y)$). Define
\[
Y_k \;=\; G_k^{-1}\!\big(U_k;V\big).
\]
Then $E[Y_k\mid V]=m_k(V)$ and $Y_k\in(0,\bar Y)$ a.s.; since the mapping uses only $V$ and $U_k$, we have $Y_k \perp Z\mid V$. By Lemma \ref{lem:lead-completeness}, for each $k$ the function $m_k$ reproduces the model-implied moment $E[Y D_k\mid Q(Z)=q]$.
\end{proof}

\begin{proof}[Proof of Theorem \ref{thm:sharpness}]
Fix any $m_k:(0,1)^J\to(0,\bar Y) \in\mathcal I_k^0$. For each $k'\neq k$, pick an admissible $m_{k'}$ that (i) satisfies its own leading-term equalities (as in \eqref{eq:lead-equality} with $k'$ in place of $k$), and (ii) satisfies the ranking inequalities jointly with $m_k$ (e.g., use pointwise envelopes based on the identified conditional MTRs for treatment $k'$ and extend them constantly in uninvolved coordinates). The collection $(m_1,\ldots,m_T)$ then meets the premises of Lemma \ref{lem:feasible-completion}, which delivers a DGP consistent with all assumptions and with $E[Y_k\mid V]=m_k$. Thus every $m_k\in\mathcal I_k^0$ is observationally equivalent to the truth under some admissible DGP.

If $m_k\notin\mathcal I_k^0$, then it violates at least one model-implied equality or inequality. Such a violation contradicts the observed moments or the assumed ordering and hence is refuted by the data under the maintained assumptions.

Therefore $\mathcal I_k^0$ is sharp.
\end{proof}

\begin{proof}[Proof of Theorem \ref{thm:sharpness_multi}]
Same logic as in the marginal case, but now we consider the full collection $\{ m_k \}_{k\in K}$ simultaneously. By Lemma 2, any such tuple in $\mathcal I_{K}^0$ can be completed into a full DGP satisfying the model and matching the observed data on all quantities used in identification. Hence, the joint identified set is sharp.
\end{proof}

\begin{proof}[Proof of Lemma~\ref{lemma:unbiased_GPRTE}]
We take the conditional expectation of $\Delta W_{N_o}$ given $\mathbf{Z}$:
\[
\mathbb{E}[\Delta W_{N_o} \mid \mathbf{Z}] = \sum_{o=1}^{N_o} \omega_o \, \mathbb{E} \left[ \hat\mu'^o_{E[Y|V]} - \hat\mu^o_{E[Y|V]} \mid Z^o \right].
\]

By construction, each $v_m^o$ is independently drawn from the distribution of $V$, and the estimators $\hat\mu^o_{E[Y|V]}$ and $\hat\mu'^o_{E[Y|V]}$ are computed by averaging over these draws. Denote $\bar D^o = d(v, Q(Z^o))$, and observe:
\begin{align*}
    \mathbb{E} \left[ \hat\mu^o_{E[Y|V]} \mid Z^o \right]
    &= \mathbb{E} \left[ \left. \frac{1}{M} \sum_{m=1}^M E[Y_{\bar D^o_m} \mid V = v_m^o, Z^o] \right| Z^o \right] \\
    &= \mathbb{E} \left[ E[Y_{\bar D^o} \mid V = v, Z^o] \mid Z^o \right] \\
    &= \sum_{k=1}^T \int d_k(v, Q(Z^o)) E[Y_k \mid V = v] f_V(v) dv,
\end{align*}
and similarly for $\mathbb{E} \left[ \hat\mu'^o_{E[Y|V]} \mid Z^o \right]$.

Substituting back into the expectation of $\Delta W_{N_o}$ gives:
\[
\mathbb{E}[\Delta W_{N_o} \mid \mathbf{Z}] 
= 
\sum_{o=1}^{N_o} \omega_o \sum_{k=1}^{T} \int \left( d'_k(v, Q'(Z^o)) - d_k(v, Q(Z^o)) \right) E[Y_k \mid V = v] f_V(v) dv,
\]
which coincides with the definition of $\Delta W$ in Equation~\eqref{eq:GPRTE}. Therefore, we conclude that $\Delta W_{N_o}$ is an unbiased estimator of $\Delta W$.
\end{proof}

\begin{proof}[Proof of Theorem \ref{thm:coveragerate}]
The coverage is driven by the Lemma 4 of \cite{imbens2004confidence}. We verify all its requirements are satisfied in our situation. 

The covariance between the lower and upper bounds is:
\[
\mathrm{Cov}(\underline{\Delta W_{N_o}}, \overline{\Delta W_{N_o}}) = \sum_{o=1}^{N_o} \omega_o^2 \frac{\mathrm{Cov}(\underline{\Delta \mu^o}, \overline{\Delta \mu^o})}{M},
\]
where \( \mathrm{Cov}(\underline{\Delta \mu^o}, \overline{\Delta \mu^o}) \) is the covariance between the sample lower and upper bounds for observation \( o \). The unbiased estimator for this covariance is:
\[
\widehat{\mathrm{Cov}}(\underline{\Delta \mu^o}, \overline{\Delta \mu^o}) = \frac{1}{M-1} \sum_{m=1}^M \left( \underline{\delta^o_m} - \underline{\Delta \mu^o_M} \right) \left( \overline{\delta^o_m} - \overline{\Delta \mu^o_M} \right).
\]
Substituting this into the covariance formula, the estimator for \( \mathrm{Cov}(\underline{\Delta W_{N_o}}, \overline{\Delta W_{N_o}}) \) is:
\[
\widehat{\mathrm{Cov}}(\underline{\Delta W_{N_o}}, \overline{\Delta W_{N_o}}) = \sum_{o=1}^{N_o} \omega_o^2 \frac{\widehat{\mathrm{Cov}}(\underline{\Delta \mu^o}, \overline{\Delta \mu^o})}{M}.
\]

The correlation coefficient is defined as
\[
\rho_{\underline{\Delta W}, \overline{\Delta W}} = \frac{\mathrm{Cov}(\underline{\Delta W_{N_o}}, \overline{\Delta W_{N_o}})}{\sigma_{\underline{\Delta W}} \sigma_{\overline{\Delta W}}},
\]
then the estimated correlation coefficient is:
\[
\widehat{\rho}_{\underline{\Delta W}, \overline{\Delta W}} = \frac{\widehat{\mathrm{Cov}}(\underline{\Delta W_{N_o}}, \overline{\Delta W_{N_o}})}{\widehat{\sigma}_{\underline{\Delta W}} \widehat{\sigma}_{\overline{\Delta W}}}.
\]

The estimators \( \widehat{\sigma}_{\underline{\Delta W}}, \widehat{\sigma}_{\overline{\Delta W}}, \widehat{\mathrm{Cov}}(\underline{\Delta W_{N_o}}, \overline{\Delta W_{N_o}}), \) and \( \widehat{\rho}_{\underline{\Delta W}, \overline{\Delta W}} \) converge to their population counterparts as \( M \to \infty \), ensuring the asymptotic validity of the estimated variance-covariance matrix.


Besides, since $\overline{\Delta W_{N_o}} - \underline{\Delta W_{N_o}}$ is bounded, so are $\widehat{\sigma}_{\underline{\Delta W}}$ and $\widehat{\sigma}_{\overline{\Delta W}}$.


Moreover, we need to verify that 
for all \(\epsilon_0 > 0\), there exist constants \(\nu_0 > 0\), \(K_0 > 0\), and \(M_0 > 0\) such that for all \(M > M_0\),
\[
\Pr\left(\sqrt{M} \left| (\overline{\Delta W_{N_o}} - \underline{\Delta W_{N_o}}) 
- (\mu_{\overline{\Delta W}} - \mu_{\underline{\Delta W}}) \right| 
> K_0(\mu_{\overline{\Delta W}} - \mu_{\underline{\Delta W}})^{\nu_0} \right) < \epsilon_0.
\]

Write $\mu_R = \mu_{\overline{\Delta W}} - \mu_{\underline{\Delta W}}$. 
We know from earlier derivations that \(\overline{\Delta W_{N_o}}\) and \(\underline{\Delta W_{N_o}}\) jointly follow a normal distribution as \(M \to \infty\), conditional on \(N_o\) being fixed. Therefore, \(R_{N_o}\) is also normally distributed:
\[
\sqrt{M} (R_{N_o} - \mu_R) \sim \mathcal{N}(0, M \sigma_R^2),
\]
where \(\sigma_R^2 = \mathrm{Var}(\overline{\Delta W_{N_o}} - \underline{\Delta W_{N_o}})\).  
Using the tail probability of the standard normal distribution \(Z \sim \mathcal{N}(0, 1)\), we know:
\[
\Pr(|Z| > z) = 2(1 - \Phi(z)),
\]
where \(\Phi(z)\) is the cumulative distribution function of the standard normal. For our case
\[
Z = \frac{\sqrt{M} (R_{N_o} - \mu_R)}{\sqrt{M} \sigma_R} = \frac{R_{N_o} - \mu_R}{\sigma_R},
\]
Then the condition to be verified becomes:
\[
\Pr\left(\frac{|R_{N_o} - \mu_R|}{\sigma_R} > \frac{K_0 \mu_R^{\nu_0}}{\sqrt{M} \sigma_R} \right) < \epsilon_0.
\]

As \(M \to \infty\), the term \(\frac{K_0 \mu_R^{\nu_0}}{\sqrt{M} \sigma_R}\) grows larger. Therefore, for any given \(\epsilon_0 > 0\), we can choose a sufficiently large \(M_0\) such that:
\[
\frac{K_0 \mu_R^{\nu_0}}{\sqrt{M} \sigma_R} \geq z_{\epsilon_0},
\]
where \(z_{\epsilon_0}\) satisfies \(2(1 - \Phi(z_{\epsilon_0})) = \epsilon_0\).

The constants \(K_0\) and \(\nu_0\) serve as scaling factors that ensure the bound grows appropriately with \(\mu_R\). For sufficiently large \(M\), \(K_0\) can be chosen proportional to \(\sigma_R\), and \(\nu_0\) determines the nonlinearity in the scaling with \(\mu_R\). These constants are less critical as \(M\) increases because \(\frac{\sqrt{M}}{K_0 \mu_R^{\nu_0}}\) dominates the tail behavior.

By choosing \(M_0\) large enough, and selecting \(K_0\) and \(\nu_0\) to ensure the scaling matches the tail decay of the normal distribution, the probability:
\[
\Pr\left(\sqrt{M} \left| R_{N_o} - \mu_R \right| > K_0 \mu_R^{\nu_0} \right)
\]
can be made arbitrarily small. Thus, the condition is satisfied.


By the Lemma 4 of \cite{imbens2004confidence}, Equation \eqref{eq:asympcover} holds.
\end{proof}

\end{document}